\documentclass{article}
\usepackage[margin=1in]{geometry}
\usepackage{natbib}  %
\usepackage{caption} %
\usepackage[colorlinks,citecolor=blue]{hyperref}
\usepackage{parskip}
\usepackage{amsmath}
\usepackage{amssymb}
\usepackage{booktabs}

\usepackage{amsfonts}       %
\usepackage{nicefrac}       %
\usepackage{microtype}      %
\usepackage{amsmath}
\usepackage{cleveref}
\usepackage{multicol}

\usepackage{complexity}

\usepackage{amsmath,amsfonts,bm}

\def\ceil#1{\lceil #1 \rceil}

\def\1{\bm{1}}

\def\eps{{\epsilon}}

\def\vmu{{\bm{\mu}}}

\def\vd{{\bm{d}}}

\def\vx{{\bm{x}}}

\DeclareMathAlphabet{\mathsfit}{\encodingdefault}{\sfdefault}{m}{sl}
\SetMathAlphabet{\mathsfit}{bold}{\encodingdefault}{\sfdefault}{bx}{n}

\renewcommand{\R}{\mathbb{R}}

\newcommand{\reg}{\mathrm{Reg}}

\usepackage{xcolor}

\usepackage{url}
\usepackage{mleftright}
\usepackage{sgame}
\usepackage{nicefrac}

\let\cite\citep

\usepackage{amsmath}
\usepackage{mathtools}

\usepackage{physics}
\usepackage{xspace}
\usepackage{bm}

\usepackage{makecell}
\usepackage{relsize}
\usepackage{stmaryrd}
\usepackage[shortlabels,inline]{enumitem}
\usepackage{colortbl}
\usepackage{multirow}
\usepackage{xcolor}
\usepackage{flafter}
\usepackage{float}
\usepackage{siunitx}
\usepackage{bbm}

\usepackage{tikz}
\usepackage{forest}
\usepackage{tikz-qtree}
\usepackage{pifont}

 \usepackage{pgfplots} 
 \pgfplotsset{compat=1.18} 
\usepackage{subfigure}
\usepackage{adjustbox}

\usetikzlibrary{backgrounds}

\let\poly\relax
\let\Root\varnothing
\let\R\relax
\let\E\relax

\newcommand{\defeq}{\coloneqq}

\newcommand{\ind}[1]{\mathbbm{1}\qty{#1}}

\newcounter{savedfootnote}

\usepackage{linegoal}
\usepackage{amsthm}
\usepackage{thm-restate}
\newcommand{\range}[1]{\llbracket #1 \rrbracket}

\usepackage{cleveref}
\usepackage{autonum}

\makeatletter
\autonum@generatePatchedReferenceCSL{Cref}
\autonum@generatePatchedReferenceCSL{cref}
\makeatother

\makeatletter
\def\NAT@spacechar{~}%
\makeatother

\DeclareMathOperator*{\E}{\mathbb E}

\let\op\operatorname
\let\eps\varepsilon
\let\mc\mathcal
\allowdisplaybreaks

\newcommand{\R}{\ensuremath{\mathbb{R}}\xspace}

\newcommand{\guideoffline}{{\sc FullFeedbackSteer}\xspace}
\newcommand{\guideonline}{{\sc OnlineSteer}\xspace}

\newcommand{\mediator}{0}
\newcommand{\chance}{{\bm{\mathsf{C}}}}

\renewcommand{\vec}{\bm}
\newcommand{\mat}{\mathbf}

\newcommand{\poly}{\op{poly}}
\newcommand{\ie}{{\em i.e.}\xspace}
\newcommand{\eg}{{\em e.g.}\xspace}

\newcommand{\Vrs}{V^{\texttt{U}}}
\newcommand{\Ers}{E^{\texttt{U}}}
\newcommand{\Grs}{G^{\texttt{U}}}
\definecolor{darkgreen}{rgb}{0,0.5,0}

\theoremstyle{plain}
\newtheorem{theorem}{Theorem}[section]
\newtheorem*{theorem*}{Theorem}
\newtheorem{proposition}[theorem]{Proposition}
\newtheorem{lemma}[theorem]{Lemma}
\newtheorem{corollary}[theorem]{Corollary}

\theoremstyle{definition}
\newtheorem{definition}[theorem]{Definition}

\theoremstyle{remark}

\makeatletter
\def\thm@space@setup{%
  \thm@preskip=\parskip \thm@postskip=0pt
}
\makeatother

\definecolor{p1color}{RGB}{31,119,180}
\definecolor{p2color}{RGB}{255,127,14}
\definecolor{p3color}{RGB}{44,160,44}
\definecolor{p4color}{RGB}{214,39,40}

\setlist{noitemsep}

\setcitestyle{authoryear}

\title{Steering No-Regret Learners to a Desired Equilibrium}

\author{%
Brian Hu Zhang\thanks{Equal contribution.} \\
Carnegie Mellon University\\
\texttt{bhzhang@cs.cmu.edu} \\
\and
Gabriele Farina$^*$ \\
MIT\\
\texttt{gfarina@mit.edu} \\
\and
Ioannis Anagnostides\\
Carnegie Mellon University\\
\texttt{ianagnos@cs.cmu.edu} \\
\and
Federico Cacciamani\\
DEIB,
Politecnico di Milano\\
\texttt{federico.cacciamani@polimi.it}\\
\and
Stephen McAleer\\
Carnegie Mellon University\\
\texttt{smcaleer@cs.cmu.edu} \\
\and
Andreas Haupt\\
MIT\\
\texttt{haupt@mit.edu} \\
\and
Andrea Celli\\
Bocconi University\\
\texttt{andrea.celli2@unibocconi.it}\\
\and
Nicola Gatti\\
DEIB,
Politecnico di Milano\\
\texttt{nicola.gatti@polimi.it} \\
\and
Vincent Conitzer\\
Carnegie Mellon University\\
\texttt{conitzer@cs.cmu.edu} \\
\and
Tuomas Sandholm \\
Carnegie Mellon University\\
Strategic Machine, Inc. \\
Strategy Robot, Inc. \\
Optimized Markets, Inc. \\
\texttt{sandholm@cs.cmu.edu} \\
}

\date{\today}

\begin{document}

\maketitle

\begin{abstract}

A mediator observes no-regret learners playing an extensive-form game repeatedly across $T$ rounds. The mediator attempts to \emph{steer} players toward some desirable predetermined equilibrium by giving (nonnegative) payments to players. We call this the \emph{steering problem}. The steering problem captures problems several problems of interest, among them {\em equilibrium selection} and {\em information design} (persuasion).  If the mediator's budget is unbounded, steering is trivial because the mediator can simply pay the players to play desirable actions. We study two bounds on the mediator's payments: a {\em total budget} and a {\em per-round budget}. If the mediator's total budget does not grow with $T$, we show that steering is impossible. However, we show that it is enough for the total budget to grow {\em sublinearly} with $T$, that is, for the {\em average} payment to vanish. When players' full strategies are observed at each round, we show that constant per-round budgets permit steering. In the more challenging setting where only trajectories through the game tree are observable, we show that steering is impossible with constant per-round budgets in general extensive-form games, but possible in normal-form games or if the per-round budget may itself depend on $T$. We also show how our results can be generalized to the case when the equilibrium is being computed {\em online} while steering is happening. We supplement our theoretical positive results with experiments highlighting the efficacy of steering in large games.
\end{abstract}

\newpage
\setcounter{tocdepth}{2} %
\tableofcontents
\newpage

\section{Introduction}

Any student of game theory learns that games can have multiple equilibria of different quality---for example, in terms of social welfare. As such, a foundational problem that has received tremendous interest in the literature revolves around characterizing the quality of the equilibrium reached under \emph{no-regret} learning dynamics. The outlook that has emerged from this endeavor, however, is discouraging: typical learning algorithms can fail spectacularly at reaching desirable equilibria. This is rather dramatically illustrated in the example of~\Cref{fig:payments} (second panel). Learning agents initialized at either $\textsf{A}$, $\textsf{B}$, or $\textsf{C}$ will in fact converge to the {\em Pareto-pessimal} Nash equilibrium of the game (bottom-left corner); only an initialization close to the Pareto-dominant equilibrium (such as \textsf{D} in the top-right corner) will end up with the desired outcome.

Our goal in this paper is to develop  methods to {\em steer} learning agents toward better equilibrium outcomes. To do so, we will use a {\em mediator} that can observe the agents playing the game, give {\em advice} to the agents (in the form of action recommendations), and {\em pay} the agents as a function of what actions they played. Our goal is to develop algorithms that allow the mediator to steer agents to a target equilibrium, while not spending too much money doing so. Critically, our only assumption on the agents' behavior is that they have no regret in hindsight. This is a fairly mild assumption compared to the assumptions made by many past papers on similar topics. We will elaborate on the comparison to related work in \Cref{sec:rel}.

Beyond the obvious relation to equilibrium selection, our model also has implications for the problem of {\em information design} and Bayesian persuasion~(\eg, \citealt{Kamenica11:Bayesian}). Indeed, we will show that we can steer players not only to any Nash equilibrium but to any {\em Bayes-correlated equilibrium (BCE)}---the solution concept most naturally associated with the problem of information design. We will also show that it is possible, in certain cases, to steer agents toward particular equilibria in an {\em online} manner, that is, {\em compute} the optimal equilibrium {\em while} steering players toward it.

\begin{figure}[!ht]
    \centering
    {\input{figures/noncredible-threat}}
    \includegraphics[width=0.75\textwidth]{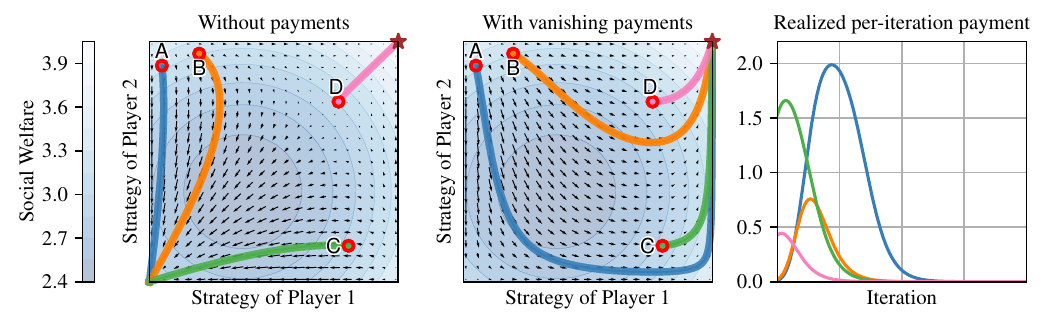}
    \vspace{-2mm}
    \caption{\textbf{Left:} An extensive-form version of a stag hunt. Chance plays uniformly at random at the root note, and the dotted line connecting the two nodes of Player 2 indicates an infoset: Player 2 cannot distinguish the two nodes. The game has two equilibria: one at the bottom-left corner, and one at the top-right corner (star). The latter is Pareto-dominant. Introducing \emph{vanishing} realized payments alters the gradient landscape, steering players to the optimal equilibrium (star) instead of the suboptimal one (opposite corner). The capital letters show the players' initial strategies. Lighter color indicates higher welfare and the star shows the highest-welfare equilibrium. Further details are in \Cref{sec:figures}.}
    \label{fig:payments}
\end{figure}

\subsection{Summary of our Results}

Here we summarize our model and results. There is a fixed, arbitrary extensive-form game $\Gamma$, being played repeatedly over rounds $t = 1, \dots, T$. Players' rewards are assumed to be normalized to range $[0,1]$. The players are assumed to play in such a way that their regret increases sublinearly as a function of $T$. This is a fairly natural and mild assumption (as discussed in the previous paragraph), and moreover there are many well-known algorithms that players can use to efficiently achieve sublinear regret in extensive-form games, perhaps the best-known of which is {\em counterfactual regret minimization}~\cite{Zinkevich07:Regret}, which has regret $T^{1/2}$ ignoring game-dependent constants.\footnote{Throughout the introduction, game-dependent constants are omitted for clarity and to emphasize the dependence on $T$. In all cases, the omitted game-dependent constant is polynomial in the number of nodes in the game tree.}

Broadly speaking, the goal of our paper is to design methods of {\em steering} the learning behavior of the players so that they reach desirable equilibria instead of undesirable ones. We do this by introducing a {\em mediator} to the game. After each round, the mediator observes how the players played the game, and has the power to give nonnegative {\em payments} $p_i^{(t)}$ to each player $i$ at each round $t$. We will first consider the case where a target {\em pure Nash equilibrium} is given as part of the problem instance.

A few observations follow easily. If the mediator's payments are not bounded, the mediator can trivially steer the players toward any outcome at all---not just equilibrium outcomes---by simply paying the players to play that outcome. We must therefore somehow bound the budget of the mediator. We will study two different budgets: a {\em per-round budget}, which constrains the individual payments $p_i^{(t)}$, and a {\em total budget}, which constrains their sum over time. We start by showing that the total budget must be allowed to grow with time.
\begin{proposition}[Informal version of \Cref{prop:imposs-boundedpayments}]
    For any fixed total budget $B$, there is a time horizon $T$ large enough that the steering problem is impossible.
\end{proposition}
As a result, the total budget must be allowed to grow with the time horizon, but yet, for the problem to be interesting, the budget cannot be allowed to grow too fast. We thus focus on the regime where the budget is allowed to grow with $T$, but only {\em sublinearly}---that is, the {\em average} per-round payment must vanish in the limit $T \to \infty$. We are interested in algorithms for which both the average budget and rate of convergence to the desired equilibrium can both be bounded by $T^{-c}$ for some absolute constant $c > 0$. We show the following.
\begin{theorem}[Informal version of \Cref{th:normalform}]
    Steering to pure-strategy equilibria is possible in normal-form games, with absolute constant per-round budget. The average budget and rate of convergence to equilibrium are both $T^{-1/4}$. 
\end{theorem}

Intuitively, the mediator sends payments in such a way as to 1) reward the player a small amount for playing the equilibrium, and 2) {\em compensate} the player for deviations of other players. The goal of the mediator is to set the payments in such a way that the target equilibrium actions become {\em strictly dominant} for the players, and therefore the players must play them.

Next we turn to the extensive-form setting. Settings such as information design, in which first a signal is designed, and then players take actions, are naturally extensive-form games. We distinguish between two settings: the {\em full feedback} setting, in which the mediator observes every player's entire strategy at every round, and the {\em trajectory-feedback setting}, in which the mediator only observes the trajectories that are actually played by the players.\footnote{This distinction becomes only meaningful for extensive-form games. For normal-form games, the two settings above coincide, because the ``trajectory'' in a normal-form game {\em is} just a list consisting of each player's chosen action.}

The {\em full feedback} setting yields results similar to the normal-form setting.
\begin{theorem}[Informal version of \Cref{th:offline-then-steer}]
    Steering to pure-strategy equilibria is possible in extensive-form games with full feedback, with absolute constant per-round budget. The average budget and rate of convergence to equilibrium are both $T^{-1/4}$. 
\end{theorem}
The {\em trajectory feedback} case, however, is quite different. 
\begin{theorem}[Informal version of \Cref{th:bandit-lower-bound}]
    With only trajectory feedback and absolute constant per-round budget, steering in general extensive-form games is impossible, even to the welfare-maximizing pure Nash equilibrium.
\end{theorem}
Intuitively, the discrepancy is because, with only trajectory feedback, it is not possible to make the target equilibrium dominant using only nonnegative, vanishing-on-average payments, so the techniques used for the previous results cannot apply. This phenomenon can already be observed in the ``stag hunt'' game in \Cref{fig:payments}: for Player 2, {\sf Stag (S)} cannot be a weakly-dominant strategy unless a payment is given at the boxed node, which would be problematic because such payments would also appear in the welfare-optimal equilibrium {\sf (S, S)}. Thus, one needs to be more clever. Fortunately, steering is still possible in this setting, but only if the per-round budget is also allowed to grow:
\begin{theorem}[Informal version of \Cref{th:offline-then-steer}]
    Steering to pure-strategy equilibria is possible in extensive-form games with full feedback. The average budget and rate of convergence to equilibrium are both $T^{-1/8}$, and the per-round budget grows at rate $T^{1/8}$. 
\end{theorem}

Next, we generalize our results beyond pure Nash equilibria. To do this, we will require the mediator to have the additional ability to give {\em advice} to the players, in the form of action recommendations. First, we show that using advice is a {\em necessary} condition for steering to even mixed Nash equilibria.
\begin{theorem}[Informal version of \Cref{th:advice-necessary}]
    Without advice, there exists a normal-form game in which the unique optimal Nash equilibrium is mixed, and it is impossible to steer players toward it. 
\end{theorem}

If we allow advice, it turns out to be possible to steer players not just to mixed Nash equilibria but to a far broader set of equilibria known as the {\em Bayes-correlated equilibria}.
\begin{theorem}[Informal version of \Cref{th:bayes-correlated}]
    With advice, steering to Bayes-correlated equilibria is possible in extensive-form games. The conditions and rates are the same as those for pure Nash equilibria.
\end{theorem}
Intuitively, the result follows because Bayes-correlated equilibria can be viewed as the pure Nash equilibria of an {\em augmented game} in which the advice is treated as part of the game's observations. Bayes-correlated equilibria are a very general solution concept that include, for example, all the extensive-form correlated equilibria~\cite{Stengel08:Extensive} and communication equilibria~\cite{Forges86:Approach,Myerson86:Multistage}, among other notions. 

Finally, we give an {\em online} version of our algorithm, which does not need to know the target equilibrium beforehand. Instead, given an objective function, the online steering algorithm {\em steers players toward the optimal equilibrium while computing it.} 
\begin{theorem}[Informal version of \Cref{th:online-steer}]
    In the full-feedback setting with advice and absolute constant per-round budget, it is possible to learn the optimal equilibrium while simultaneously steering the players toward it. The average budget and rate of convergence to equilibrium are both $T^{-1/6}$. 
\end{theorem}
As before, in normal-form games, full feedback and trajectory feedback essentially coincide, so online steering also turns out to be possible in normal-form games with trajectory feedback. In extensive-form games, however, the problem of trajectory-feedback online steering seems more difficult, and we leave it as an open problem. We summarize the rates we obtain in~\Cref{tab:steering-summary}.

Finally, we complement our theoretical analysis by implementing and testing our steering algorithms in several benchmark games in \Cref{sec:experiments}.

\begin{table}[!t]
    \centering
    \caption{Summary of our positive algorithmic results. We hide game-dependent constants and logarithmic factors, and assume that regret minimizers incur regret $T^{-1/2}$.}
    \label{tab:steering-summary}
    \begin{tabular}{c|cc}
      & Steering to Fixed Equilibrium  & Online Steering \\\toprule
     Normal Form or Full Feedback & $T^{-1/4}$ (\Cref{th:offline-then-steer}) & $T^{-1/6}$ (\Cref{th:online-steer}) \\
    Extensive Form and Trajectory Feedback & $T^{-1/8}$ (\Cref{th:bandit-offline-steer}) & {\em Open problem} \\ \bottomrule
    \end{tabular}
    \vspace{-4mm}
\end{table}

\subsection{Related Work}
\label{sec:rel}

\paragraph{$k$-implementation}

Our setting and our algorithms are closely related to the problem of {\em $k$-implementation}~\cite{Monderer04:Implementation} in normal-form games (see also~\citep{Deng16:Complexity} for pertinent complexity considerations). In $k$-implementation, the goal is to make a certain strategy profile a {\em (weakly) dominant strategy} for all players using nonnegative payments. \citet{Monderer04:Implementation} observe that only Nash equilibria can be implemented using zero realized payments. Our \guideoffline algorithm operates in a similar setting: by precomputing an equilibrium and giving payments in such a way that players are {\em sandboxed}, each player's dominant strategy is to be direct, so the players converge. Indeed, for {\em normal-form} games, the fact that all pure Nash equilibria are (in the language of $k$-implementation) $0$-implementable implies that steering is possible for normal-form games, in both the full-feedback and online settings. Our \guideoffline and {\sc TrajectorySteer} algorithms could then be interpreted as saying that arbitrary Nash equilibria of extensive-form games can be implemented in unique normal-form coarse correlated equilibria (and therefore the unique convergence point of no-regret learning dynamics). 

However, our results are more general than  \cite{Monderer04:Implementation} in several ways. First, our rationality assumption differs. Instead of considering players that play weakly-dominant strategies, we consider a no-regret assumption. We show that in extensive form, this distinction is meaningful: it is sometimes impossible to make the desirable equilibrium weakly dominant, and this leads to a more intricate proof for \Cref{th:bandit-offline-steer}. Second, we consider a wider class of games: Our algorithms work in arbitrary extensive-form settings, not just normal form. As we discussed above, this allows for a natural formulation of information design problems and a theoretically rich problem in the trajectory-feedback setting. Even in the full-information setting, working in extensive form means that we need to be careful in designing the payment scheme so that the maximum possible payment $P$ is constant.  For instance, \Cref{th:bandit-lower-bound} shows that no absolute constant payment can suffice. Finally, we make less restrictive information assumptions.Algorithm $\guideonline$ {\em learns the equilibrium while steering agents toward it}.

\paragraph{Steering to near-optimal equilibria}

Moreover, ample of prior research has endeavored to steer strategic agents toward ``good'' equilibria~\citep{Mguni19:Coordinating,Li20:End,Kempe20:Inducing,Liu22:Inducing,Dinh21:Last,Bishop21:How}. Indeed, the presence of a centralized party that can help ``nudge'' behavior to a better state has served as a central motivation for the literature on the \emph{price of stability}~\citep{Anshelevich08:The,Schulz03:On,Agussurja09:The,Panageas16:Average}, thereby allowing to circumvent impossibility results in terms of the worst Nash equilibria~\citep{Koutsoupias99:Worst,Roughgarden05:Selfish}. For example, as articulated by~\citet{Balcan09:Improved}: ``In cases where there are both high and low cost Nash equilibria, a central authority could hope to `move' behavior from a high-cost equilibrium to a low-cost one by running a public service advertising campaign promoting the better behavior.'' Nevertheless, \citet{Balcan09:Improved} also stress that it is unrealistic to assume that all agents blindly follow the prescribed protocol, unless it is within their interest to do so; this is indeed a key motivation for our considerations. \citet{Balcan11:Leading,Balcan13:Circumventing,Balcan14:Near} also endeavor to lead learning dynamics to a desired state for certain classes of games, although there are key differences between those papers and our setting. In particular, focusing on the work of~\citet{Balcan14:Near} for concreteness, our paper shows that steering is possible under the mild assumption that players have vanishing average regret, while \citet{Balcan14:Near} impose much stronger behavioral assumptions; namely, in the first phase of their protocol players who receive advise are assumed to obey, even though it may not be in their own interest, while the rest of the players are following best response dynamics. Further, while in the protocol of~\citet{Balcan14:Near} advise is provided to a subset of the players, they only guarantee convergence to an \emph{approximately} optimal state; by contrast, our focus here is on steering to optimal equilibria.

On a related direction, \citet{Kleinberg11:Beyond} identify a class of games where specific learning dynamics lead to much better social welfare compared to the Nash equilibrium. Relatedly, Roughgarden's smoothness framework~\citep{Roughgarden15:Intrinsic} gives bounds on the (time-average) social welfare guarantees under no-regret learners, but imposes somewhat restrictive assumptions on the underlying class of games.

\paragraph{Strategizing against no-regret learners} 

Our problem of steering no-regret learners to desirable outcomes is also connected to the problem of strategizing against no-regret learners, studied from different perspectives in several prior papers~\citep{Deng19:Strategizing,Kolumbus21:How,Freeman20:No,Roughgarden17:Online,Dandrea23:Playing,Cho21:Machine,Mansour22:Strategizing,Brown23:Is,Li23:Exploiting,Cai23:Selling}. Particularly relevant is \citet{Deng19:Strategizing}. The paper considers the choice of strategies against a single no-regret agent, and asks the question of whether outcomes better than Stackelberg/mechanism design outcomes, in which the agent optimally responds to actions by the mediator are achievable (with a negative answer). The paper assumes that the game does not have weakly dominated pure strategies for the no-regret agent, and that the learner is mean-based~\cite{braverman2018selling}. Our setting is more general in game class (all extensive-form games) and in terms of the power of the principal (payments and advice, in contrast to actions in a game).

Moreover, introducing nonnegative payments to incentivize specific outcomes bears resemblance to the setting of \emph{contract design}~\citep{Duetting22:Multi,Dutting21:The,Dutting21:Combinatorial,Guruganesh24:Contracting}, and has been recently employed in \emph{federated learning} as well to encourage participation (\emph{e.g.}, \citealp{Hu23:Federated}). Finally, our study relates to the literature of mechanism design that adopts vanishing regret as a behavioral assumption \cite{camara2020mechanisms,braverman2018selling,Fikioris23:Liquid}.

\section{Preliminaries}\label{sec:prel}

In this section, we introduce some basic background on extensive-form games.

\begin{definition}
    An \emph{extensive-form game} $\Gamma$ with $n$ players has the following components:
    \begin{enumerate}%
        \item a set of players, identified with the set of integers $\range{n} := \{ 1, \dots, n \}$. We will use $-i$, for $i \in \range{n}$, to denote all players except $i$;
        \item a directed tree $H$ of {\em histories} or {\em nodes}, whose root is denoted $\Root$. The edges of $H$ are labeled with {\em actions}. The set of actions legal at $h$ is denoted $A_h$. Leaf nodes of $H$ are called {\em terminal}, and the set of such leaves is denoted by $Z$;
        \item a partition $H \setminus Z = H_\chance \sqcup H_1 \sqcup \dots \sqcup H_n$, where $H_i$ is the set of nodes at which $i$ takes an action, and $\chance$ denotes the chance player;
        \item for each player $i \in \range{n}$, a partition $\mc I_i$ of $i$'s decision nodes $H_i$ into {\em information sets}. Every node in a given information set $I$ must have the same set of legal actions, denoted by $A_I$;\looseness-1
        \item for each player $i$, a {\em utility function} $u_i : Z \to [0,1]$ which we assume to be bounded; and
        \item for each chance node $h \in H_\chance$, a fixed probability distribution $c(\mathop{\cdot}|h)$ over $A_h$.
    \end{enumerate}
\end{definition}

At a node $h \in H$, the {\em sequence} $\sigma_i(h)$ of an agent $i$ is the set of all information sets encountered by agent $i$, and the actions played at such information sets, along the $\Root \to h$ path, excluding at $h$ itself. An agent has {\em perfect recall} if $\sigma_i(h) = \sigma_i(h')$ for all $h, h'$ in the same infoset. Unless otherwise stated (\Cref{se:mediator}), we assume that all players have perfect recall. We will use $\Sigma_i := \{ \sigma_i(z) : z \in Z \}$ to denote the set of all sequences of player $i$ that correspond to terminal nodes. 

A {\em pure strategy} of player $i$ is a choice of one action in $A_I$ for each information set $I \in \mc I_i$. The {\em sequence form} of a pure strategy is the vector $\vec x_i \in \{0, 1\}^{\Sigma_i}$ given by $\vec x_i[\sigma] = 1$ if and only if $i$ plays every action on the path from the root to sequence $\sigma \in \Sigma_i$. We will use the shorthand $\vec x_i[z] = \vec x_i[\sigma_i(z)]$. A {\em mixed strategy} is a distribution over pure strategies, and the sequence form of a mixed strategy is the corresponding convex combination $\vec x_i \in [0, 1]^{\Sigma_i}$. We will use $X_i$ to denote the polytope of sequence-form mixed strategies of player $i$.

A profile of mixed strategies $\vec x = ( \vec x_1, \dots, \vec x_n) \in X := X_1 \times \dots \times X_n$, induces a distribution over terminal nodes. We will use $z \sim \vec x$ to denote sampling from such a distribution. The expected utility of agent $i$ under such a distribution is given by $u_i(\vec x) := \E_{z \sim  \vec x} u_i(z)$. Critically, the sequence form has the property that each agent's expected utility is a linear function of its own sequence-form mixed strategy. For a profile $\vec x \in X$ and set $N \subseteq \range{n}$, we will use the notation $\hat{\vx}_N \in \R^Z$ to denote the vector $\hat{\vx}_N[z] = \prod_{j \in N} \vx_j[z]$, and we will write $\hat{\vx} := \hat{\vx}_{\range{n}}$. A Nash equilibrium is a strategy profile $\vec x$ such that, for any $i \in \range{n}$ and any $\vx_i' \in X_i$,
$u_i (\vec x) \ge u_i (\vx_i', \vec x_{-i}).$

\section{The Steering Problem}
\label{sec:decentralization}

In this section, we introduce what we call the {\em steering} problem. Informally, the steering problem asks whether a mediator can always steer players to any given equilibrium of an extensive-form game.

\begin{definition}[Steering Problem for Pure-Strategy Nash Equilibrium]
    \label{def:dec-steering}
Let $\Gamma$ be an extensive-form game with payoffs bounded in $[0,1]$. Let $\vec d$ be an arbitrary pure-strategy Nash equilibrium of $\Gamma$, which we will call the {\em target equilibrium}. The mediator knows the game $\Gamma$, as well as a function $R(T) = o(T)$, which may be game-dependent, that bounds the regret of all players. At each round $t \in \range{T}$, the mediator picks {\em payment functions} for each player, $p_i^{(t)} : X_1 \times \cdots \times X_n \to [0, P]$, where $p_i^{(t)}$ is linear in $\vx_i$ and continuous in $\vx_{-i}$, and $P$ defines the largest allowable per-iteration payment. Then, players pick strategies $\vx_i^{(t)} \in X_i$. Each player $i$ then gets utility $v_i^{(t)}(\vx_i) := u_i(\vx_i, \vx^{(t)}_{-i}) + p_i^{(t)}(\vx_i, \vx^{(t)}_{-i})$. The mediator has two desiderata.
    \begin{enumerate}[(S1)]
        \item \label{item:payments} (Payments) The time-averaged realized payments to the players, defined as 
                $$
                  \max_{i \in [n]} \frac{1}{T} \sum_{t=1}^T  p_i^{(t)}(\vx^{(t)}),$$ 
        converges to 0 as $T \to \infty$.
        \item \label{item:opt-med} (Target Equilibrium) Players' actions are indistinguishable from the Nash equilibrium $\vec d$. That is, for every terminal node $z$, the {\em directness gap}, defined as
        $$\sum_{z \in Z} \abs{\frac{1}{T} \sum_{t=1}^T \hat{\vx}^{(t)}[z] - \hat{\vec d}[z]} = \norm{\frac{1}{T} \sum_{t=1}^T \hat{\vx}^{(t)} - \hat{\vec d}}_1,$$ 
        converges to $0$ as $T \to \infty$.
    \end{enumerate}
\end{definition}
The assumption imposed on the payment functions in \Cref{def:dec-steering} ensures the existence of Nash equilibria in the payment-augmented game (\emph{e.g.}, \citealp{Fudenberg91:Game}, p. 34). Throughout this paper, we will refer to players as {\em direct} if they are playing actions prescribed by the target equilibrium strategy $\vec d$. Critically, \ref{item:opt-med} does not require that the strategies themselves converge to the direct strategies, \ie, $\vx_i^{(t)} \to \vec d_i$, in iterates or in averages. They may differ on nodes off the equilibrium path.
Instead, the requirement defined by \ref{item:opt-med} is that the {\em outcome distribution over terminal nodes} converges to that of the equilibrium. Similarly, \ref{item:payments} refers to the {\em realized} payments $p_i^{(t)}(\vec x^{(t)})$, not the {\em maximum offered payment} $\max_{\vec x \in X} p_i^{(t)}(\vec x)$. 

For now, we assume that the pure Nash equilibrium is part of the instance, and therefore our only task is to steer the agents toward it. In \Cref{se:mediator} we show how our steering algorithms can be extended to other equilibrium concepts such as {\em mixed} or {\em (Bayes-)correlated} equilibria, and to the case where the mediator needs to compute the equilibrium.

The mediator does not know anything about how the players pick their strategies, except that they will have regret bounded by a function that vanishes in the limit and is known to the mediator. This condition is a commonly adopted behavioral assumption~\citep{Nekipelov15:Econometrics,Kolumbus22:Auctions,camara2020mechanisms}. The regret of Player $i \in \range{n}$ in this context is defined as\looseness-1
\begin{equation}
    \label{eq:reg-payment}
    \reg_{X_i}^T \defeq \frac{1}{P+1} \qty[ \max_{\vx_i^* \in X_i}  \sum_{t=1}^T v_i^{(t)}(\vx_i^*) - \sum_{t=1}^T v_i^{(t)}(\vx_i^{(t)})]. %
\end{equation}
That is, regret takes into account the payment functions offered to that player. (The division by $1/(P+1)$ is for  normalization, since $v_i^{(t)}$s has range $[0, P+1]$.) The assumption of bounded regret is realistic even in extensive-form games, as  various regret minimizing algorithms exist. Two notable examples are the {\em counterfactual regret minimization} (CFR) framework~\cite{Zinkevich07:Regret}, which yields {\em full-feedback} regret minimizers, and IXOMD~\cite{Kozuno21:Learning} (see also~\citep{Fiegel23:Adapting,Bai22:Near}), which yields {\em bandit-feedback} regret minimizers.

How large payments are needed to achieve \ref{item:payments} and \ref{item:opt-med}? If the mediator could provide totally unconstrained payments, it could enforce any arbitrary outcome. On the other hand, if the total payments are restricted to be bounded, the steering problem is information-theoretically impossible:

\begin{restatable}{proposition}{boundpay}
    \label{prop:imposs-boundedpayments}
    There exists a game and some function $R(T) = O(\sqrt{T})$ such that, for all $B \ge 0$, the steering problem is impossible if we add the constraint $\sum_{t=1}^\infty \sum_{i = 1}^n  p_i^{(t)}(\vx^{(t)}) \le B$.
\end{restatable}

\begin{proof}
    Suppose that the mediator's goal is for the players to coordinate on the equilibrium ({\sf B, B}) in the coordination $2$-player game with the following payoff matrix.
    \begin{center}
        \begin{tabular}{ccc}
              & \sf A & \sf B \\\toprule
           \sf  A & 0.5, 0.5 & 0,0 \\\midrule
           \sf  B & 0,0 & 1,1 \\ \bottomrule
        \end{tabular}
    \end{center}
    Set $R(T) = 2\sqrt{T}$. We will show that, regardless of the mediator's strategy, it is possible for the players to play ({\sf A, A}) for all but finitely many rounds.
    
    Suppose the players play as follows. Let $\Gamma^{(t)}$ be the game at time $t$ induced by the mediator's payoff function $p^{(t)}$. For the first $B^2$ rounds, play an arbitrary Nash equilibrium of $\Gamma^{(t)}$. After that, if ({\sf A, A}) is a Nash equilibrium of $\Gamma^{(t)}$, play it. Otherwise, play a strategy profile $\vx^{(t)}$ for which $\sum_{i=1}^n p_i^{(t)}(\vx^{(t)}) > \frac{1}{2}$ (Such a strategy profile must exist, for otherwise ({\sf A, A}) would be a Nash equilibrium).
    
    The total regret of the players after $T$ rounds is (at most) $0$ for $T \leq B^2$, since we have assumed that they are playing a Nash equilibrium of $\Gamma^{(t)}$, and at most $(P+1)k$ for $T > B^2$, where $k$ is the number of times that the final case triggers, since the reward range of $\Gamma^{(t)}$ is at most $[0, P+1]$. But the final case can only trigger at most $2 B$ times since the mediator only has a total budget of $B$. Therefore, the regret is bounded by $2 (P+1) \sqrt{T}/(P+1) = 2\sqrt{T}$ for any $T$, and for all but $2 B + B^2$ rounds, the players are playing a suboptimal equilibrium. So, desideratum~\ref{item:opt-med} in~\Cref{def:dec-steering} cannot be satisfied.
\end{proof}

Hence, a weaker requirement on the size of the payments is needed. Between these extremes, one may allow the \emph{total} payment to be unbounded, but insist that the \emph{average} payment per round must vanish in the limit.

\section{Steering in Normal-Form Games}
\label{sec:nfgs}
We start with the simpler setting of {\em normal-form games}, that is, extensive-form games in which every player has one information set, and the set of histories correspond precisely to the set of pure profiles. This setting is much simpler than the general extensive-form setting (we consider in the next section), and we can appeal to a special case of a result in the literature\citet{Monderer04:Implementation}.
\begin{proposition}[Costless implementation of pure Nash equilibria, special case of $k$-implementation, \citealp{Monderer04:Implementation}]
    Let $\vec d$ be a pure Nash equilibrium in a normal-form game. Then there exist functions $p_i^* : X_1 \times \dots \times X_n \to [0, 1]$, with $p_i^*(\vec d) = 0$, such that in the game with utilities $v_i := u_i + p_i^*$, the profile $\vec d$ is weakly dominant:
    $v_i(\vec d_i, \vec x_{-i}) \ge v_i(\vec x_i, \vec x_{-i})$ for every profile $\vec x$.
\end{proposition}
The proof is constructive. The payment function 
\begin{align}
    p_i^*(\vec x) := (\vec d_i^\top \vec x_i) \qty\Big(1 - \prod_{j \ne i} \vec d_j^\top \vec x_j),
\end{align}
which on pure profiles $\vec x$ returns $1$ if and only if $\vec x_i = \vec d_i$ and $\vec x_j \ne \vec d_j$ for some $j \ne i$ makes equilibrium play weakly dominant. It is {\em almost} enough for steering: the only problem is that $\vec d$ is only {\em weakly} dominant, so no-regret players \emph{may} play other strategies than $\vec d$. This can be fixed by adding a small reward $\alpha \ll 1$ for playing $\vec d_i$. That is, we set
\begin{align}\label{eq:normal-form-steer}
    p_i(\vec x) := \alpha \vec d_i^\top \vec x_i + p_i^*(\vec x) = (\vec d_i^\top \vec x_i)\qty\Big(\alpha + 1 - \prod_{j \ne i} \vec d_j^\top \vec x_j).
\end{align}

On a high level, the structure of the payment function guarantees that the average strategy of any no-regret learner $i \in \range{n}$ should be approaching the direct strategy $\vec{d}_i$ by making $\vec{d}_i$ the strictly dominant strategy of player $i$. At the same time, it is possible to ensure that the average payment will also be vanishing by appropriately selecting parameter $\alpha$. With an appropriate choice of $\alpha$, this is enough to solve the steering problem for normal-form games:

\begin{restatable}[Normal-form steering]{theorem}{thNormalform}\label{th:normalform}
    Let $p_i(\vec x)$ be defined as in \eqref{eq:normal-form-steer}, set $\alpha = \sqrt{\eps}$, where $\eps := 4n R(T) / T$, and let $T$ be large enough that $\alpha \le 1$. Then players will be steered toward equilibrium, with both payments and directness gap bounded by $2\sqrt{\eps}$. 
\end{restatable}

\begin{proof}
    By construction of the payments, the utility for player $i$ is at least $\alpha$ higher for playing $\vec d_i$ than for any other action, regardless of the actions of the other players. Let $\eps := nR(T)/T$ and $\delta_i^{(t)} := 1 - \vec d_i^\top \vec x_i^{(t)}$. Then the above property ensured by the payments implies that
    $
        R(T)/T = \eps/n\ge \alpha \E_{t \in \range{T}} \delta_i^{(t)}.
    $
    Let $z^*$ be the terminal node induced by profile $\vec d$. Then the directness gap is
    \begin{align}
        2 \E_{t} \qty [1 - \hat{\vx}^{(t)}[z^*]] = 2 - 2 \E_{t} \prod_i (1 - \delta_i^{(t)}) \le 2 \E_t \sum_i \delta_i^{(t)} \le 2\eps / \alpha,
    \end{align}
    and the payments are bounded by 
    \begin{align}
        \E_t p_i(\vec x) \le \alpha + \E_t (1 - \prod_{j \ne i} (1 - \delta_i^{(t)})) \le \alpha + \eps / \alpha.
    \end{align}
    So, taking $\alpha = \sqrt{\eps}$ completes the proof.
\end{proof}

We note that no effort was made throughout this paper to optimize the game-dependent or constant factors, so long as they remained polynomial in $|Z|$---they can very likely be improved.

\section{Steering in Extensive-Form Games}
\label{sec:steering-EFGs}
This section considers steering in extensive-form games. We will first consider a model in which steering payments can condition on full player strategies (\Cref{subsec:full.feedback}). Next, we consider a model in which only realized trajectories are considered (\Cref{sec:steering-bandit}).

Tbere are two main reassons why the extensive-form version of the steering problem is  significantly more challenging than the normal-form version.

First, in extensive form, the strategy spaces of the players are no longer simplices. Therefore, if we wanted to write a payment function $p_i$ with the property that $p_i(\vec x) = \alpha \ind{\vec x = \vec d} + \ind{\vec x_i = \vec d_i; \exists j\, \vec x_j \ne \vec d_j}$ for pure $\vec x$ (which is what was needed by \Cref{th:normalform}), such a function would not be linear (or even convex) in player $i$'s strategy $\vec x_i \in X_i$ (which is a sequence-form strategy, not a distribution over pure strategies). As such, even the meaning of extensive-form regret minimization becomes suspect in this setting.

Second, in extensive form, a desirable property would be that the mediator give payments conditioned only on what actually happens in gameplay, {\em not} on the players' full strategies---in particular, if a particular information set is not reached during play, the mediator should not know what action the player {\em would have} selected at that information set. We will call this the {\em trajectory} setting, and distinguish it from the {\em full-feedback} setting, where the mediator observes the players' full strategies.\footnote{To be clear, the settings are differentiated by what the {\em mediator} observes, not what the {\em players} observe. That is, it is valid to consider the full-feedback steering setting with players running bandit-feedback regret minimizers, or the trajectory-feedback steering setting with players running full-feedback regret minimizing algorithms.} This distinction is meaningless in the normal-form setting: since terminal nodes in normal form correspond to (pure) profiles, observing gameplay is equivalent to observing strategies. (We will discuss this point in more detail when we introduce the trajectory-feedback setting in \Cref{sec:steering-bandit}.)

\subsection{Steering with Full Feedback}\label{subsec:full.feedback}

In this section, we introduce a steering algorithm for extensive-form games under full feedback, summarized below.

\begin{definition}[\guideoffline]
    At every round, set the payment function $p_i(\vx_i, \vx_{-i})$ as
\begin{equation}
\label{eq:payment}
    \underbrace{\alpha \vec d_i^\top \vx_i \vphantom{\min_{\vx_i'}}}_\text{directness bonus}  + \underbrace{\qty[u_i(\vx_i, \vec d_{-i}) - u_i(\vx_i, \vx_{-i})]\vphantom{\min_{\vx_i'}}}_\text{sandboxing payments} - \underbrace{\min_{\vx_i' \in X_i} \qty[u_i(\vx_i', \vec d_{-i}) - u_i(\vx_i', \vx_{-i})],}_\text{payment to ensure nonnegativity}
\end{equation}
where $\alpha \le 1/|Z|$ is a hyperparameter that we will select appropriately. 
\end{definition}

By construction, $p_i$ satisfies the conditions of the steering problem (\Cref{def:dec-steering}): it is linear in $\vx_i$, continuous in $\vx_{-i}$, nonnegative, and bounded by an absolute constant (namely, $3$).
The payment function defined above has three terms:
\begin{enumerate}[leftmargin=5.5mm]
    \item The first term is a {\em reward for directness}: a player gets a reward proportional to $\alpha$ if it plays $\vec d_i$.
    \item The second term {\em compensates the player} for the indirectness of other players. That is, the second term ensures that players' rewards are {\em as if} the other players had acted directly.
    \item The final term simply ensures that the overall expression is nonnegative.
\end{enumerate}
We claim that this protocol solves the basic version of the steering problem, as formalized below. 

\begin{restatable}{theorem}{offlinesteer}
    \label{th:offline-then-steer}
    Set $\alpha = \sqrt{\eps},$  where $\eps := 4n R(T) / T$, and let $T$ be large enough that $\alpha \le 1/|Z|$. Then, \guideoffline results in average realized payments and directness gap at most $3|Z| \sqrt{\eps}$.
\end{restatable}

Before proving \Cref{th:offline-then-steer}, we start by stating a useful lemma, which is proven in~\Cref{sec:proofs-steering}.

\begin{restatable}{lemma}{aux}
    \label{lem:obedience-union}
    Let $\bar{\vx}_i \defeq \E_{t \in \range{T}} \vx_i^{(t)}$ for any player $i \in \range{n}$ and $\delta := \sum_{i=1}^n \vec d_i^\top (\vec d_i - \bar{\vx}_i)$. Then, $
    \E_{t \in \range{T}} \norm{\hat{\vx}_N^{(t)} - \hat{\vec d}_N}_1  \le |Z| \delta
$ for every $N \subseteq \range{n}$. Moreover, if the payments are defined according to \eqref{eq:payment}, the average payment to every player can be bounded by $\E_{t \in \range{T}} p_i(\vec x^{(t)}) \le |Z|(2\delta + \alpha)$.
\end{restatable}

\begin{proof}[Proof of~\Cref{th:offline-then-steer}]
The utility of each player $i \in \range{n}$ reads
\begin{equation}
     v_i(\vx_i, \vx_{-i}) \defeq \alpha \Vec{d}_i^\top \Vec{x}_i + u_i(\vx_i, \vec{d}_{-i}) - \min_{\vx_i' \in X_i} [ u_i(\vx_i', \vec{d}_{-i}) - u_i(\vx_i', \vec{x}_{-i})].
\end{equation}
Given that $\vec d$ is an equilibrium, it follows that $\vec{d}_i$ is a strict best response for any player $i \in \range{n}$. That is, the regret of each player $i \in \range{n}$ after $T$ iterations can be lower bounded as
\begin{equation}
    \sum_{t=1}^T \left( \alpha \Vec{d}_i^\top (\Vec{d}_i - \vx_i^{(t)}) + u_i(\vec{d}) - u_i(\vec{x}_i^{(t)}, \Vec{d}_{-i}) \right) \geq \alpha T \vec d_i^\top (\vec d_i - \bar{\vx}_i),
\end{equation}
where we used that $u_i(\vec{d}) - u_i(\vec{x}_i^{(t)}, \Vec{d}_{-i}) \geq 0$ since $\vec d$ is an equilibrium. Thus, \begin{align}\sum_{i=1}^n \vec d_i^\top (\vec d_i - \bar{\vx}_i) \leq \frac{n R(T)}{\alpha T} = \frac{\eps}{\alpha}.\end{align} We can now apply \Cref{lem:obedience-union} to obtain that $\E_{t \in \range{T}} \norm{\hat{\vx}^{(t)} - \hat{\vec d}}_1  \le |Z| \delta$, where $\delta := \eps / \alpha$. Thus, the directness gap is bounded by
\begin{align}
    \norm{\E_t \hat{\vec x}^{(t)} - \hat{\vec d}}_1 = \E_t \norm{\hat{\vx}^{(t)} - \hat{\vec d}}_1 \le \frac{n|Z|R(T)}{\alpha T},
\end{align}
where the first equality follows because $\hat{\vec{d}}$ is an extreme point of $X$ (as in~\eqref{align:extreme}). Furthermore, by \Cref{lem:obedience-union}, the payment to each player $i \in \range{n}$ can be bounded by 
\begin{align} 2|Z|(2\delta + \alpha) = 2|Z|\frac{\eps}{\alpha} + |Z| \alpha.\end{align}
As a result, setting $\alpha = \sqrt{\eps}$ for $T$ sufficiently large so that $\alpha \le 1/|Z|$, we guarantee that the payment to each player is bounded by $3n|Z| \sqrt{\eps}$ and the directness gap is bounded by $|Z| \sqrt{\eps}$,
as desired.
\end{proof}

\subsection{Steering in the Trajectory-Feedback Setting}
\label{sec:steering-bandit}

In \guideoffline, payments depend on full strategies $\vec x$, not the realized game trajectories. In particular, the mediator in \Cref{th:offline-then-steer} observes what the players {\em would have played} even at infosets that other players avoid. To allow for an algorithm that works without knowledge of full strategies, $p_i^{(t)}$ must be structured so that it could be induced by a payment function that only gives payments for terminal nodes reached during play. To this end, we now formalize {\em trajectory-feedback steering}.

\begin{definition}[Trajectory-feedback steering problem]\label{def:banditsteerproblem}
    Let $\Gamma$ be an extensive-form game in which rewards are bounded in $[0,1]$ for all players. Let $\vec d$ be an arbitrary pure-strategy Nash equilibrium of $\Gamma$. The mediator knows $\Gamma$ and a regret bound $R(T) = o(T)$. At each $t \in \range{T}$, the mediator selects a payment function $q_i^{(t)} : Z \to [0, P]$. The players select strategies $\vec x_i^{(t)}$. A terminal node $z^{(t)} \sim \vec x^{(t)}$ is sampled, and all agents observe the terminal node that was reached, $z^{(t)}$. The players get payments $q_i^{(t)}(z^{(t)})$, so that their expected payment is $p_i^{(t)}(\vec x) := \E_{z \sim \vec x} q_i^{(t)}(z)$. The desiderata are as in \Cref{def:dec-steering}.
\end{definition}

The trajectory-feedback steering problem is more difficult than the full-feedback steering problem in two ways. First, as discussed above, the mediator does not observe the strategies $\vec x$, only a terminal node $z^{(t)} \sim \vec x$.
Second, the form of the payment function $q_i^{(t)} : Z \to [0, P]$ is restricted: this is already sufficient to rule out \guideoffline. Indeed, $p_i$ as defined in \eqref{eq:payment} cannot be written in the form $\E_{z \sim  \vec x} q_i(z)$: $p_i(\vx_i, \vx_{-i})$ is nonlinear in $\vx_{-i}$ due to the nonnegativity-ensuring payments, whereas every function of the form  $\E_{z \sim \vec x} q_i(z)$ will be linear in each player's strategy. 

We remark that, despite the above algorithm containing a sampling step, the payment function is defined {\em deterministically}: the payment is defined as the {\em expected value} $p_i^{(t)}(\vec x) := \E_{z \sim \vec x} q_i^{(t)}(z)$. Thus, the theorem statements in this section will also be deterministic.

In the normal-form setting, the payments $p_i$ defined by \eqref{eq:normal-form-steer} already satisfy the condition of trajectory-feedback steering. In particular, if $z$ is the terminal node, we have
$$
    p_i(\vec x) = \E_{z \sim \vec x} \qty[ \alpha \ind{z = z^*} + \ind{\vec x_i = \vec d_i; \exists j\, \vec x_j \ne \vec d_j} ].
$$
Therefore, in the normal-form setting, \Cref{th:normalform} applies to both full-feedback steering and trajectory-feedback steering, and we have no need to distinguish between the two. However, in extensive form, as discussed above, the two settings are quite different.

\subsubsection{Lower bound}

Unlike in the full-feedback or normal-form settings, in the trajectory-feedback setting, steering is impossible in the general case in the sense that per-iteration payments bounded by any constant do not suffice. 
\begin{restatable}{theorem}{banditlowerbound}\label{th:bandit-lower-bound}
    For every $P > 0$, there exists an extensive-form game $\Gamma$ with $O(P)$ players, $O(P^2)$ nodes, and rewards bounded in $[0, 1]$ such that, with payments $q_i^{(t)} : Z \to [0, P]$, it is impossible to steer players to the welfare-maximizing Nash equilibrium, even when $R(T) = 0$.
\end{restatable}
 For intuition, consider the extensive-form game in \Cref{fig:proof combined}, which can be seen as a three-player version of Stag Hunt. Players who play {\sf Hare (H)} get a value of $1/2$ (up to constants); in addition, if all three players play {\sf Stag (S)}, they all get expected value $1$. The welfare-maximizing equilibrium is ``everyone plays {\sf Stag}'', but ``everyone plays {\sf Hare}'' is also an equilibrium. In addition, if all players are playing {\sf Hare}, the only way for the mediator to convince a player to play {\sf Stag} without accidentally also paying players in the {\sf Stag} equilibrium is to pay players at one of the three boxed nodes. But those three nodes are only reached with probability $1/n$ as often as the three nodes on the left, so the mediator would have to give a bonus of more than $n/2$. The full proof essentially works by deriving an algorithm that the players could use to exploit this dilemma to achieve either large payments or bad convergence rate, generalizing the example to $n > 3$, and taking $n = \Theta(P)$. We next formalize this intuition.

 \begin{proof}
For any $n > 0$, consider the following $n$-player extensive-form game $\Gamma$, which has $O(n^2)$ nodes. Every player has only a single information set with two actions, and we will (for good reason, as we will see later) refer to the actions as {\sf Stag} and {\sf Hare}. Chance first picks some $j \in \range{n} \cup \{ \bot \}$ uniformly at random. 

If $j \ne \bot$, then player $j$ plays an action (which is either {\sf Stag} or {\sf Hare}). If $i$ plays {\sf Hare}, it gets utility $1/2$; otherwise, it gets utility $0$. All other players get utility $0$.

If $k = \bot$, chance samples another player $k$ uniformly at random from $\range{n}$. Then, in the order $k, k+1, \dots, n, 1, 2, \dots, k-1$, the players play their actions. If any player at any point plays {\sf Hare}, then the game ends and all players get $0$. If all players play {\sf Stag}, then all players get $1$.

The normal form of this game is an $n$-player generalization of the Stag Hunt game: if all players play {\sf Stag} then all players have (expected) payoff $1/(n+1)$; if any player plays {\sf Hare} then every player has expected payoff $(1/2)/(n+1)$ for playing {\sf Hare} and $0$ for playing {\sf Stag}. In particular, the welfare-optimal profile, ``everyone plays {\sf Stag}'', is a Nash equilibrium, and hence is also the welfare-optimal EFCE,  with social welfare $n/(n+1)$. ``Everyone plays {\sf Hare}'' is also an equilibrium, with social welfare $(1/2)n/(n+1)$. The game tree when $n=3$ is depicted in \Cref{fig:proof combined}. 

Intuitively, the rest of the proof works as follows. Suppose that all players are currently playing {\sf Hare}. The mediator needs to incentivize players to play {\sf Stag}, but it has a dilemma. It cannot give a large payment to $i$ for playing {\sf Stag} when $j = i$---then the average payment for each player would diverge if the players were to move to the {\sf Stag} equilibrium. The only other location that the mediator could possibly give a payment to $i$ is when $j = \bot$, $k = i$, player $i$ plays {\sf Stag}, and the next player plays {\sf Hare}. But this node is only reached with probability $O(1/n^2)$---therefore, to outweigh $i$'s current incentive of $\Theta(1/n)$ of playing {\sf Hare}, the payment at this node would have to be $\Theta(n)$, at which point taking $n = \Theta(P)$ would complete the proof.

We now formalize this intuition. Take $n = \ceil{4P}.$ Consider players who play as follows. At each timestep $t$, the players consider the extensive-form game $\Gamma^{(t)}$ induced by adding the payment functions $q^{(t)}_i$ that the mediator would play, and ignoring mediator recommendations. That is, $\Gamma^{(t)}$ is identical to $\Gamma$ except that $q^{(t)}_i$ has been added to player $i$'s utility function. If ``everyone plays {\sf Hare}'' is a Nash equilibrium of $\Gamma^{(t)}$, all players play {\sf Hare}. Otherwise, the players play according to an arbitrary Nash equilibrium of $\Gamma^{(t)}$. 

Since the players are playing according to a Nash equilibrium at every step, they all have regret at most $0$. Now consider two cases.
\begin{enumerate}
    \item There is a player $i$ such that plays {\sf Stag} with probability less than $1/2$. Then the social welfare is at most $(3/4) n/(n+1)$, which is lower than the optimal social welfare by $(1/4)n/(n+1)$. 
    \item All players play {\sf Stag} with probability at least $1/2$. Then, in particular, ``everyone plays {\sf Hare}'' is not a Nash equilibrium in $\Gamma^{(t)}$. So, if everyone were to play {\sf Hare}, there is some player $i$ who would rather deviate and play {\sf Stag}. Thus, the mediator must be giving an expected payment to $i$ of at least $(1/2)/(n+1)$. As discussed above, there are only two nodes $z$ for which the setting of $q_i^{(t)}(z)$ increases $i$'s utility for playing {\sf Stag} relative to its utility for playing {\sf Hare}. The first is when $j = \bot$, $k = i$, $i$ plays {\sf Stag}, and the next player plays {\sf Hare}. Since $P \le n/4$ and this node occurs with probability $1/(n(n+1))$, even the maximum payment at this node contributes at most $(1/4)/(n+1)$ to the expected payment. Therefore, the remainder of the payment, $(1/2)/(n+1)$, must be given when $j = i$ and then $i$ plays {\sf Stag}. But Player $i$ plays {\sf Stag} with probability at least $1/2$, so $i$'s observed expected payment is at least $(1/4)/(n+1)$. 
\end{enumerate}
Therefore, we have
\begin{align}
    \qty(u_\mediator^* - \E u_0(z^{(t)})) + \E \sum_{i \in \range{n}} q_i^{(t)}(z^{(t)})\ge \frac{1}{4(n+1)}
\end{align}
where $u_0$ is the social welfare function, so it is impossible for both quantities to tend to $0$ as $T \to \infty$.
\end{proof}

\begin{figure}[!ht]
    \input{figures/combined_proof_figures}
\end{figure}

\subsubsection{Upper bound}
\label{se:bandit-offline-steer}

To circumvent the lower bound in \Cref{th:bandit-lower-bound}, in this subsection, we allow the payment bound $P \ge 1$ to depend on both the time limit $T$ and the game. %

\begin{definition}[\textsc{TrajectorySteer}]
    \label{def:banditsteer}
    Let $\alpha, P$ be hyperparameters. Then, for all rounds $t = 1, \dots, T$, sample $z \sim \vx^{(t)}$ and pay players as follows.
    If all players have been direct (\ie, if $\hat{\vd}[z] = 1$), pay all players $\alpha$.
     If at least one player has not been direct, pay $P$ to all players who have been direct.
That is, set $q_i^{(t)}(z^{(t)}) = \alpha \hat{\vd}[z] + P \vd_i[z] (1 - \hat{\vd}[z])$.
\end{definition}

\begin{restatable}{theorem}{banditofflinesteer}
    Set the hyperparameters
    $
        \alpha = 4|Z|^{1/2} \eps^{1/4}$ and $P = {2|Z|^{1/2}}{\eps^{-1/4}},
    $
    where $\eps := R(T) / T$, and let $T$ be large enough that $\alpha \le 1$. Then, running \textsc{TrajectorySteer} for $T$ rounds results in average realized payments bounded by
    $ 8 |Z|^{1/2} \eps^{1/4}$, and directness gap by $2\eps^{1/2}$.
    \label{th:bandit-offline-steer}
\end{restatable}

As alluded to in the introduction, the proof of this result is more involved than those for previous results, because one cannot simply make the target equilibrium dominant as in the full-feedback case. One may hope that---as in \guideoffline---the desired equilibrium can be made dominant by adding payments.
In fact, a sort of ``chicken-and-egg'' problem arises: \ref{item:opt-med}  requires that all players converge to equilibrium. But for this to happen, other players' strategies must first converge to equilibrium so that $i$'s incentives are as they would be in equilibrium. The main challenge in the proof of \Cref{th:bandit-offline-steer} is therefore to carefully set the hyperparameters to achieve convergence despite these apparent problems.

\section{Other Equilibrium Notions and Online Steering}\label{se:mediator}
So far, \Cref{th:offline-then-steer,th:bandit-offline-steer} handle only the case where the equilibrium is a \emph{pure-strategy} Nash equilibrium of the game, given as part of the input. This section extends our analysis to other equilibrium notions and considers settings in which an {\em objective for the mediator} is given instead of a target equilibrium. For the former, we will show that many types of equilibrium can be viewed as pure-strategy equilibria in an {\em augmented game} in which the mediator has the ability to give {\em advice} to the players in the form of action recommendations. Then, in the original game, the goal is to guide the players to the pure strategy profile of following recommendations.

\subsection{Necessity of Advice}\label{app:recommendations}

We first show that without the possibility to give advice, steering is impossible with sublinear payments.
\begin{theorem}\label{th:advice-necessary}
    There exists a normal-form game, and objective function $u_0$ of the mediator, such that the unique optimal equilibrium is mixed, and it is impossible to steer players toward that equilibrium using only sublinear payments (and no advice).
\end{theorem}
\begin{proof}
    Consider a 2-player, binary action coordination game, with actions {\sf A} and {\sf B}. Players receive utility $1$ point for playing the same action, and $-1$ otherwise. The mediator's goal is to {\em minimize} the welfare of the players.\footnote{One could construct an example in which the mediator's goal is to {\em maximize} the players' utility, by simply adding a third player, with one action, whose utility is $-10$ if P1 and P2 play the same action and $0$ otherwise.}

    The welfare-minimizing equilibrium in this game is the fully-mixed one. So, we claim that, using sublinear payments alone, it is impossible to steer players to the mixed equilibrium. Consider the following algorithm for the players: Let $\Gamma^{(t)} $ be the game at time $t$ induced by the mediator's payoff function $p^{(t)}$. Play an arbitrary Nash equilibrium of $\Gamma^{(t)}$, pure if possible.
    The total regret of the players after $T$ rounds is at most $0$ since the players always play a Nash equilibrium. There are three cases:
    \begin{enumerate}
        \item The players play ({\sf A, A}) or ({\sf B, B}). In this case, the players get social welfare $2$.
        \item The players play ({\sf A, B}) or ({\sf B, A}). In this case, the  players get social welfare $-2$ in the game itself, but in order for either of these to be a Nash equilibrium, there must be a payment of at least $2$ to each player.
        \item The players play a mixed strategy. This means that $\Gamma^{(t)}$ had no pure strategy Nash equilibrium. Since ({\sf A, A}) is not an equilibrium, suppose WLOG that $v_1^{(t)}({\sf B, A}) > v_1^{(t)}({\sf A, A})$. Then $p_i^{(t)}({\sf B, A}) > 2$. Since ({\sf B, A}) is also not a Nash equilibrium, we have $v_2^{(t)}({\sf B, B}) > v_2^{(t)}({\sf B, A})$. Since ({\sf B, B}) is also not a Nash equilibrium, we have $v_1^{(t)}({\sf A, B}) > v_1^{(t)}({\sf B, B})$, so $p_1^{(t)}({\sf A, B}) > 2$. Thus, all four strategy profiles have either high welfare for the players, or nontrivial payments.
    \end{enumerate}
    In all three cases, as a result, we must have $\sum_i u_i(\vec x^{(t)}) + 2 p_i^{(t)}(\vec x^{(t)}) > 1$ for all timesteps $t$. Therefore, summing over $t = 1, \dots, T$, it is impossible for both quantities to grow sublinearly in $T$, which is what would be required for successful steering.
\end{proof}
Given this result, we will analyze a setting in with the mediator is allowed to provide ``advice,'' and show a broad possibility result for steering.

\subsection{More General Equilibrium Notions: Bayes-Correlated Equilibrium}
Throughout this subsection, there will be two games: the original game $\hat\Gamma$, and the augmented game $\Gamma$. We will use hats to distinguish the various components of them. For example, a history of $\hat \Gamma$ is $\hat h \in \hat H$, a strategy of Player $i$ is $\hat \vx_i \in  \hat X_i$, and so on. Given an $n$-player game $\hat\Gamma$, the {\em mediator-augmented game} $\Gamma$ is the $n+1$-player game constructed as follows. $\Gamma$ is identical to $\hat\Gamma$, except that there is an extra player, namely, the mediator itself. We will denote the mediator as Player $0$. For each (non-chance) player $i$, every decision point $\hat h \in \hat H_i$ is replaced with the following gadget. First, the mediator selects an action $\hat a \in \hat A_{\hat h}$ to {\em recommend} to Player $i$. Player $i$ privately observes the recommendation, and only then is allowed to choose an action. The mediator is assumed to have perfect information in the game. To ensure that the size of $\Gamma$ is not too large, we make the following restriction: once two players have disobeyed action recommendations (``deviated''), the mediator ceases to give further action recommendations. Finally, upon reaching a terminal node $\hat z \in \hat Z$, each player gets utility $\hat u_i(\hat z)$.

We first analyze the size of $\Gamma$. A terminal node in $\Gamma$ can be uniquely identified by a tuple $(\hat z, \hat h_1, \hat h_2, \hat a_1, \hat a_2)$ where $\hat z$ is the terminal node in the original game that was reached,  $\hat h_1,\hat h_2$ are predecessors of $\hat z$ at which players deviated (or $\Root$ if the deviations did not happen), and $\hat a_1$ and $\hat a_2$ are the recommendations that the mediator gave at $\hat h_1, \hat h_2$ respectively (again, $\Root$ if the deviations did not happen). Thus, a (very loose) bound on the number of terminal nodes in $\Gamma$ is $| Z| \le |\hat Z|^3$, \ie, it is polynomial. (This is where we use the fact that only two deviations were allowed.)

As in the previous section, the mediator is able to {\em commit} to a strategy $\vmu \in  X_0$ upfront on each iteration. For a fixed mediator strategy $\vmu$, we will use $ \Gamma^{ \vmu}$ to refer to the $n$-player game resulting from treating the mediator as a nature player that plays according to $\vmu$. 

The {\em direct strategy} $\vd_i \in  X_i$ of each player $i$ is the strategy that follows all mediator recommendations. The goal of the mediator is to find a {\em Bayes-correlated equilibrium}, which is defined as follows.

\begin{definition}
    A {\em Bayes-correlated equilibrium} $\Gamma$ is a strategy $\vmu \in  X_0$ for the mediator such that $\vd$ is a Nash equilibrium of $ \Gamma^{\vmu}$. An equilibrium $\vmu$ is {\em optimal} if, among all equilibria, it maximizes the mediator's objective $ u_0(\vmu, \vd)$.
\end{definition}

Bayes-correlated equilibria (BCEs) were introduced first by \citet{bce} in single-step games. In sequential (extensive-form) games, BCEs were explored first, to our knowledge, by \citet{Makris23:Information} in the economics literature, and in independent work in the computer science literature as a special case of the general framework introduced by \citet{Zhang22:Polynomial}. Bayes-correlated equilibria are easily seen to be a superset of most other equilibrium notions, including (mixed) Nash equilibria, {\em extensive-form correlated equilibria} (EFCE)~\cite{Stengel08:Extensive}, {\em communication equilibria}~\cite{Forges86:Approach,Myerson86:Multistage}, and many more. The {\em revelation principle} assures us that the assumption that players will be direct in equilibrium is without loss of generality: for every possible Nash equilibrium $\vx$ of $ \Gamma^{\vmu}$, then there is some $\vmu'$ such that $ u_i(\vmu', \vd)= u_i(\vmu, \vx)$. 

BCEs naturally capture the problems of {\em information design} and {\em Bayesian persuasion}~(\eg, \citet{Kamenica11:Bayesian}). In particular, the results in this section can therefore be thought of as a version of information design/Bayesian persuasion that does not need to assume that players will play a certain profile ($\vd)$, but instead {\em steers} the players to play that profile.

Since $ \Gamma^{\vmu}$ is just an $n$-player game with pure Nash equilibrium $\vd$, all of the results in the previous sections apply. Therefore, it follows immediately that is possible to steer players toward {\em any} BCE (and thus any mixed Nash equilibrium, any EFCE, or any communication equilibrium) so long as the mediator is allowed to give advice to the players. We therefore have the following result.

\begin{theorem}
    \label{th:bayes-correlated}
    Algorithms \guideoffline and {\sc TrajectorySteer} can be used to steer players to an arbitrary Bayes-correlated equilibrium, with (up to a polynomial loss in the dependence on $|\hat Z|$, because $|Z| = \poly(|\hat Z|)$) the same bounds.
\end{theorem}

\subsection{Online Steering}

We now consider the setting where the target equilibrium is {\em not} given to us beforehand. We assume that the mediator wishes to steer players toward an {\em optimal} equilibrium, but does not {\em a priori} know what that optimal equilibrium is. Instead of a target Nash equilibrium, we assume that the mediator has a utility function $\hat u_0 : \hat Z \to [0, 1]$, and we will call $\hat u_0$ the {\em objective}. As with players' utility functions, $\hat u_0$ in $\hat\Gamma$ induces a mediator utility function $u_0$ in $\Gamma$. In particular, we would like to steer players toward an {\em optimal} equilibrium $\vec\mu$, without knowing that equilibrium beforehand. To that end, we add a new criterion.

\begin{enumerate}
\item[(S3)] \label{item:opt-online} (Optimality) The mediator's reward should converge to the reward of the optimal equilibrium. That is, the {\em optimality gap} $u_0^* - \frac{1}{T} \sum_{t=1}^T u_0(\vec\mu^{(t)}, \vx^{(t)})$, where $u_0^*$ is the mediator utility in an optimal equilibrium, converges to $0$ as $T \to \infty$. 
\end{enumerate}

Since equilibria in mediator-augmented games are just strategies $\tilde\vmu$ under which $\tilde\vd$ is a Nash equilibrium, we may use the following algorithm to steer players toward an optimal Bayes-correlated equilibrium:

\begin{definition}[{\sc ComputeThenSteer}]
    \label{def:computethensteer}
        Compute an optimal equilibrium $\vec\mu$. With $\vec\mu$ held fixed, run any steering algorithm in $\Gamma^{\vec\mu}$.
\end{definition}

As observed earlier, the main weakness of {\sc ComputeThenSteer} is that it must compute an equilibrium offline. Although this can be done in polynomial time~\cite{Zhang22:Polynomial}, it is still far less efficient than, for example, a single step of a regret minimier. To sidestep this, in this section we will introduce algorithms that compute the equilibrium in an {\em online} manner, while steering players toward it.
Our algorithms will make use of a Lagrangian dual formulation analyzed by~\citet{Zhang23:Computing}.

\begin{proposition}[\citet{Zhang23:Computing}]
    \label{prop:zerosum}
    There exists a (game-dependent) constant $\lambda^* \ge 0$ such that, for every $\lambda \ge \lambda^*$, the solutions $\vec\mu$ to
\begin{align}
     \max_{\vec\mu \in X_0} \min_{\vec x \in X} u_0(\vec\mu, \vec d)- \lambda \sum_{i=1}^n \qty[u_i(\vec\mu, \vec x_i, \vec d_{-i}) - u_i(\vec\mu, \vec d_i, \vec d_{-i})],\label{eq:saddle-point}
\end{align}
are exactly the optimal equilibria of the augmented game.
\end{proposition}

\begin{definition}[\guideonline]
    \label{def:online-steer}
The mediator runs a regret minimization algorithm $\mc R_\mediator$ over its own strategy space $X_0$, which we assume has regret at most $R_\mediator(T)$ after $T$ rounds. On each round, the mediator does the following:
\begin{itemize}[noitemsep]
    \item Get a strategy $\vmu^{(t)}$ from $\mc R_\mediator$. Play $\vmu^{(t)}$, and set $p_i^{(t)}$ as defined in \eqref{eq:payment} in $\Gamma^{\vec\mu^{(t)}}$.
     \item Pass utility $\vec\mu \mapsto \frac{1}{\lambda} u_0(\vec\mu, \vec d)- \sum_{i=1}^n \qty[u_i(\vec\mu, \vec x_i^{(t)}, \vec d_{-i}) - u_i(\vec\mu, \vec d_i, \vec d_{-i})] $ to $\mc R_\mediator$, where  $\lambda \ge 1$ is a hyperparameter.
\end{itemize}
\end{definition}

\begin{restatable}{theorem}{onlinesteer}
    \label{th:online-steer}
    Set the hyperparameters %
    $
        \alpha = \eps^{2/3}|Z|^{-1/3}$ and $\lambda = { |Z|^{2/3}}{\eps^{-1/3}},
    $
    where $\eps := (R_\mediator(T) + 4nR(T)) / T$ is the average regret bound summed across players, and let $T$ be large enough that $\alpha \le 1/|Z|$.
    Then running \guideonline results in average realized payments, directness gap, and optimality gap all bounded by $7 \lambda^* |Z|^{4/3} \eps^{1/3}$.
\end{restatable}
The argument now works with the zero-sum formulation~\eqref{eq:saddle-point}, and leverages the fact that the agents' average strategies are approaching the set of Nash equilibria since they have vanishing regrets. Thus, each player's average strategy should be approaching the direct strategy, which in turn implies that the average utility of the mediator is converging to the optimal value, analogously to \Cref{th:offline-then-steer}. We provide the formal argument below.

\begin{proof}[Proof of~\Cref{th:online-steer}]
To simplify the notation, we assume without loss of generality that $u^*_\mediator = 0$. We will also use the change of variables $\vec y := \vec x - \vec d \in Y := X - \vec d$. With the payments and utility functions as specified, the losses given to the players and the mediator are, up to additive constants, exactly the losses that they would see if they were playing the zero-sum game
\begin{align}\label{eq:online-steering-game}
        \max_{\vec\mu \in \Xi} \min_{\vec y \in Y}~ \frac{1}{\lambda} \vec c^\top \vec\mu - \vec\mu^\top \mat A \vec y - \alpha \vec d^\top \vec y,
\end{align}
where $\mat A = \mqty[\mat A_1 &\cdots& \mat A_n]$, $\vec d = \mqty[\vec d_1^\top & \cdots & \vec d_n^\top]^\top $, and $\lambda \ge 1$. Now let $(\lambda^*, \vec y^*)$ be an optimal dual solution in \eqref{eq:saddle-point}. If we select $\lambda \ge \lambda^*$ and $\vec y' := (\lambda^* / \lambda)\vec y^*$, then $(\lambda, \vec y')$ is also an optimal dual solution in \eqref{eq:saddle-point}. Therefore,
    \begin{align}
        \max_{\vec\mu \in \Xi} \min_{\vec y \in Y} ~ \frac{1}{\lambda} \vec c^\top \vec\mu -  \vec\mu^\top \mat A \vec y - \alpha \vec d^\top \vec y \le -\alpha \vec d^\top \vec y',
    \end{align}
    since it is assumed that $u^*_\mediator = 0$. Further, we know that $(\bar\vmu, \bar{\vec y})$ is an $\eps$-Nash equilibrium of the above zero-sum game since $(R_\mediator(T) + 4nR(T))/T = \eps$; in particular, we have that\footnote{A technical comment here: $-\vec d^\top \vec y$ is {\em nonnegative}, and takes its {\em minimum} value at $\vec y = 0$.}
    \begin{align}
      - \alpha \vec d^\top \bar{\vec y} \le \max_{\vec\mu \in \Xi} ~ \frac{1}{\lambda} \vec c^\top \vec\mu -  \vec\mu^\top \mat A \bar{\vec y} - \alpha \vec d^\top \bar{\vec y} \le -\alpha \vec d^\top \vec y' + \eps, 
    \end{align}
    or, rearranging,
    \begin{align}
        -\vec d^\top \bar{\vec y} \le -\frac{\lambda^*}{\lambda}\vec d^\top \vec y^* + \frac{\eps}{\alpha} \le \frac{\lambda^*}{\lambda} |Z| + \frac{\eps}{\alpha} := \delta.
    \end{align}
    Thus, by \Cref{lem:obedience-union}, the average payment is bounded by $|Z|(2\delta + \alpha)$.
    We now turn to the mediator's average utility. The equilibrium value of \eqref{eq:online-steering-game} is at least $-\alpha|Z|$ (achieved by the optimal equilibrium), in turn implying that the current value in the game under $(\Bar{\vmu}, \Bar{\vec{y}})$ is at least $-\alpha|Z| - \eps$. So,
    \begin{align}
        \E_{t \in \range{T}} u_\mediator(\vmu^{(t)}, \vec d) = \vec c^\top \bar{\vmu}\ge \min_{\vec y \in Y}~ \vec c^\top \bar{\vec\mu} - \lambda \qty[\bar{\vec\mu}^\top \mat A \vec y - \alpha \vec d^\top \vec y] \ge -\lambda(\alpha|Z| + 2\eps).
    \end{align}
    By \Cref{lem:obedience-union} again,
\begin{align}
    \abs{\E_{t \in \range{T}} u_\mediator(\vec\mu^{(t)}, \vec x^{(t)}) - u_\mediator(\vec\mu^{(t)}, \vec d)} \le \norm{\E_{t \in \range{T}} \hat{\vec x}^{(t)} - \hat{\vec d}}_1 \le \E_{t \in \range{T}} \norm{\hat{\vx}^{(t)} - \hat{\vec d}}_1 \le |Z|\delta,
\end{align}
so the optimality gap is bounded by $2\eps\lambda + |Z| \alpha \lambda  + |Z| \delta$, and the directness gap is bounded by $|Z|\delta$. It thus suffices to select hyperparameters $\alpha$ and $\lambda$ so as to minimize the following expression, which is an upper bound on all three gaps:
\begin{align}
  2\eps\lambda + |Z| \alpha \lambda  + 2 |Z| \delta = 2\eps\lambda + |Z| \alpha\lambda + 2|Z|^2 \frac{\lambda^*}{\lambda} + 2|Z| \frac{\eps}{\alpha}.
\end{align}
In particular, setting the hyperparameters as in the theorem statement and plugging them into the expression above, we arrive at the bound
\begin{align}
    2 \eps^{2/3} |Z|^{2/3} + |Z|^{4/3} \eps^{1/3} + 2\lambda^* |Z|^{4/3} \eps^{1/3} + 2|Z|^{4/3} \eps^{1/3} \le 7 \lambda^*  |Z|^{4/3} \eps^{1/3},
\end{align}
as claimed.
\end{proof}

It is worth noting that, despite the fact that it would speed up the convergence, we cannot set $\lambda$ and $\alpha$ dependent on $\lambda^*$, because we do not know $\lambda^*$ {\em a priori}.

Algorithm \guideonline can also be used to steer to optimal equilibria in other notions of equilibrium, such as {\em communication equilibrium}~\cite{Forges86:Approach,Myerson86:Multistage}, by using appropriate constructions of mediator-augmented games. The Bayes-correlated equilibrium is the most natural and general of these notions, so it is the one we use in our paper. For a more general discussion of mediator-augmented games, see~\citet{Zhang22:Polynomial}.

\guideonline has a further guarantee that \guideoffline does not, owing to the fact that it learns an equilibrium online: it works even when the players' sets of deviations, $X_i$, is not known upfront. In particular, the following generalization of \Cref{th:online-steer} follows from an identical proof. %

\begin{corollary}
    Suppose that each player $i$, unbeknownst to the mediator, is choosing from a subset $Y_i \subseteq X_i$ of strategies that includes the direct strategy $\vec d_i$. Then, running \Cref{th:online-steer} with the same hyperparameters yields the same convergence guarantees, except that the mediator's utility converges to its optimal utility against the {\em true} deviators, that is, a solution to \eqref{eq:saddle-point} with each $X_i$ replaced by $Y_i$.
\end{corollary}

At this point, it is very reasonable to ask whether it is possible to perform {\em online} steering with {\em trajectory} feedback. In {\em normal-form} games, as with offline setting, there is minimal difference between the trajectory- and full-feedback settings. This intuition carries over to the trajectory-feedback setting: {\sc OnlineSteer} can be adapted into an online trajectory-feedback steering algorithm for normal-form games, with essentially the same convergence guarantee. We defer the formal statement of the algorithm and proof to \Cref{app:bandit-online-nf}.\looseness-1

The algorithm, however, fails to extend to the {\em extensive-form} online trajectory-feedback setting, for the same reasons that the {\em offline} full-feedback algorithm fails to extend to the online setting. We leave extensive-form online trajectory-feedback steering as an interesting open problem.

\section{Experimental Results}
\label{sec:experiments}

We ran experiments with our {\sc TrajectorySteer} algorithm (\Cref{def:banditsteer}) on various notions of equilibrium in extensive-form games, using the {\sc ComputeThenSteer} framework suggested by 
\Cref{def:computethensteer}. Since the hyperparameter settings suggested by \Cref{def:banditsteer} are very extreme, in practice we fix a constant $P$ and set $\alpha$ dynamically based on the currently-observed gap to directness. We used CFR+ \cite{Tammelin14:Solving} as the regret minimizer for each player, and precomputed a welfare-optimal equilibrium with the LP algorithm of \citet{Zhang22:Polynomial}. In most instances tested, a small constant $P$ (say, $P \le 8$) is enough to steer CFR+ regret minimizers to the exact equilibrium in a finite number of iterations. Two plots exhibiting this behavior are shown in \Cref{fig:body-experiments}. More experiments, as well as descriptions of the game instances tested, can be found in \Cref{app:experiments}.

\begin{figure}[!hb]

    \includegraphics[width=.346\textwidth]{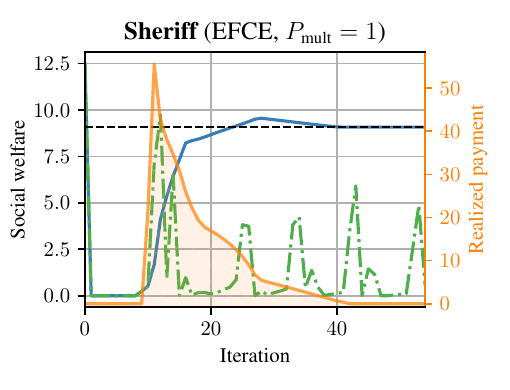}%
    \hskip-2mm%
    \includegraphics[width=.346\textwidth]{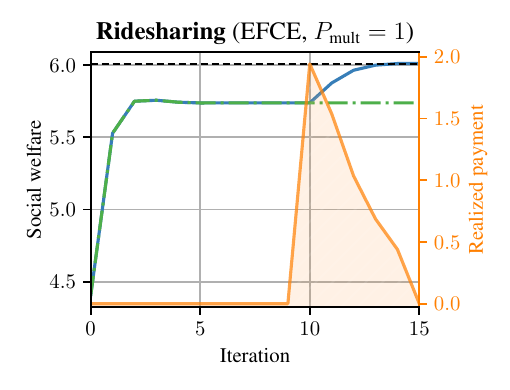}%
    \hskip-4mm%
    \includegraphics[width=.346\textwidth]{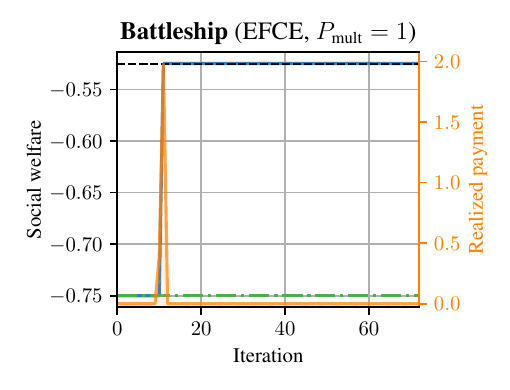}%
    \caption{Sample experimental results. The blue line in each figure is the social welfare (left y-axis) of the players {\em with} steering enabled. The green dashed line is the social welfare {\em without} steering. The yellow line  gives the payment (right y-axis) paid to each player. The flat black line denotes the welfare of the optimal equilibrium. The panels show the game, the equilibrium concept (in this figure, always EFCE). In all cases, the first ten iterations are a ``burn-in'' period during which no payments are issued; steering only begins after that.}\label{fig:body-experiments}
\end{figure}

\section{Conclusions and Future Research}

We established that it is possible to steer no-regret learners to optimal equilibria using vanishing rewards, even under trajectory feedback. There are many interesting avenues for future research. First, is there a natural {\em trajectory-feedback, online} algorithm that combines the desirable properties of both \guideonline and {\sc TrajectorySteer}? Second, this paper did not attempt to provide optimal rates, and their improvement is a fruitful direction for future work. Third, are there algorithms with less demanding knowledge assumptions for the principal, \emph{e.g.}, steering without knowledge of utility functions? Finally, our main behavioral assumption throughout this paper is that the regret players incur vanishes in the limit. Yet, stronger guarantees could be possible when specific no-regret learning dynamics are in place, such as mean-based learning~\citep{braverman2018selling}; see~\citep{Vlatakis-Gkaragkounis20:No,Giannou21:Rate,Giannou21:Survival} for recent results in the presence of \emph{strict} equilibria. Concretely, it would be interesting to understand the class of learning dynamics under which the steering problem can be solved with a finite cumulative budget.

\section*{Acknowledgements}

The work of Prof. Sandholm's research group is funded by the National Science Foundation under grants IIS1901403, CCF-1733556, and the ARO under award W911NF2210266. McAleer is funded by NSF grant \#2127309 to the Computing Research Association for the CIFellows 2021 Project. The work of Prof. Gatti's research group is funded by the FAIR (Future Artificial Intelligence Research) project, funded by the NextGenerationEU program within the PNRR-PE-AI scheme (M4C2, Investment 1.3, Line on Artificial Intelligence). Conitzer thanks the Cooperative AI Foundation and Polaris Ventures (formerly the Center for Emerging Risk Research) for funding the Foundations of Cooperative AI Lab (FOCAL). Andy Haupt was supported by Effective Giving. Andrea Celli is supported by the European Union - Next Generation EU, component M4.C2, investment 1.1. - CUP: J53D23007170001. We thank Dylan Hadfield-Menell for helpful conversations.

\bibliographystyle{plainnat}
\bibliography{dairefs}

\newpage
\appendix
\setlist{itemsep=\parskip,parsep=\parskip}

\section{Details on Figure \ref{fig:payments}}
\label{sec:figures}

In this section, we elaborate on \Cref{fig:payments}, and we provide some further pertinent illustrations. As shown in \Cref{th:bandit-lower-bound}, this is a challenging instance for steering no-regret learners in the trajectory-feedback setting. The results illustrated in \Cref{fig:payments} correspond to each player employing multiplicative weights update (MWU) under full feedback with learning rate $\eta \defeq 0.1$.

Furthermore, we also experiment with each player using a variant of EXP3 \citep{Auer02:Nonstochastic} with exploration parameter $\epsilon \defeq 5\%$. We employ our steering algortihm in the trajectory-feedback setting with different \emph{potential} payments $P$ and parameter $\alpha = 0$, leading to the results illustrated in~\Cref{fig:furtherpayments}.

\begin{figure}[H]
    \centering
    \includegraphics[width=\textwidth]{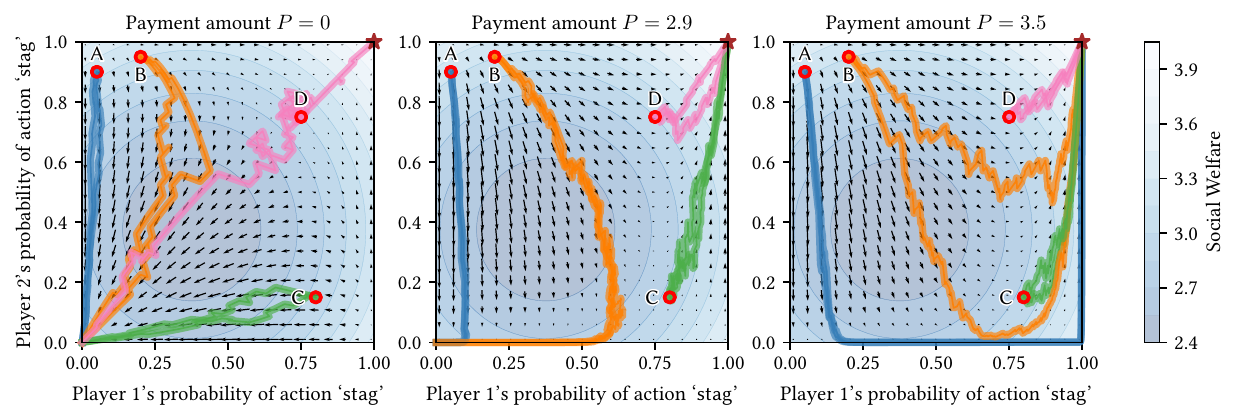}
    \caption{The trajectories of EXP3 algorithms under different random initializations and vanishing payments. Trajectories with the same color correspond to the same initialization but under different realizations of the players' sampled actions.}
    \label{fig:furtherpayments}
\end{figure}

\section{Omitted Proofs}
\label{sec:proofs-steering}

In this section, we provide the proofs omitted from the main body.

\subsection{Proof of Lemma~\ref{lem:obedience-union}}

\aux*

\begin{proof}
    Let $\delta_i := \vec d_i^\top (\vec d_i - \bar{\vx}_i)$ for any player $i \in \range{n}$. Then, we have that 
    $$\min_{z : \vec d_i[z] = 1} \vec {\bar x}_i[z] \ge 1- \delta_i, $$ 
    which in turn implies that
$\max_{z : \vec d_i[z] = 0} \vec {\bar x}_i[z] \le \delta_i.$ Now let $N \subseteq \range{n}$. If $z \in Z$ is such that $\vec d_N[z] = 1$,
\begin{align}
    \bar{\vx}_N[z] = \E_{t \in \range{T}} \vx_N^{(t)}[z] = \E_{t \in \range{T}} \prod_{j \in N} \vx_j^{(t)}[z] \ge \E_{t \in \range{T}} \prod_{j \in N} \qty(1 - \delta_j) \ge 1 - \sum_{j \in N} \delta_j = 1 - \delta.
\end{align}
Further, if $\vec d_j[z] = 0$ for some $j\in N$,
\begin{align}
    \bar{\vx}_N[z] \le \bar{\vx}_j[z] \le \delta_j \le \delta.
\end{align}
Thus,
\begin{align}
    \E_{t \in \range{T}} \norm{\hat{\vx}_N^{(t)} - \hat{\vec d}_N}_1 &= \E_{t \in \range{T}} \left( \sum_{z : \hat{\vec{d}}_N[z] = 0} (\hat{\vec{x}}^{(t)}_N[z] - \hat{\vec{d}}_N[z] ) + \sum_{z : \hat{\vec{d}}_N[z] = 1} (\hat{\vec{d}}_N[z] - \hat{\vec{x}}^{(t)}_N[z]) \right) \notag \\
    &= \norm{ \E_{t \in \range{T}} \hat{\vx}_N^{(t)} - \hat{\vec d}_N}_1 = \norm{ \bar{\vx}_N - \hat{\vec{d}}_N}_1 \le |Z| \delta, \label{align:extreme}
\end{align}
since we have shown that $| \Bar{\vx}_N[z] - \hat{\vec{d}}_N[z]| \leq \delta$ for any $z \in Z$. This establishes the first part of the claim. Next, the average payments~\eqref{eq:payment} can by bounded for any player $i \in \range{n}$ as
\begin{multline}
    \E_{t \in \range{T}} \left[\qty[u_i(\vx_i^{(t)}, \vec d_{-i}) - u_i(\vx_i^{(t)}, \vx_{-i}^{(t)})] \right. \\- \left.\min_{\vx_i' \in X_i} \qty[u_i(\vx_i', \vec d_{-i}) - u_i(\vx_i', \vx_{-i}^{(t)})] + \alpha \vec d_i^\top \vx_i^{(t)}\right]
    \\\le 2 \E_{t \in \range{T}} \norm{\hat{\vec x}_{-i}^{(t)} - \hat{\vec d}_{-i}}_1 + \alpha |Z| \le |Z|(2\delta + \alpha),
\end{multline}
where we used the normalization assumption $|u_i(\cdot)| \leq 1$, and the fact that $\vec{d}_i^\top \vx_i^{(t)} \leq |Z|$. This concludes the proof.
\end{proof}

\subsection{Proof of Theorem~\ref{th:bandit-offline-steer}}

Next, we provide the proof of \Cref{th:bandit-offline-steer}.

\banditofflinesteer*

We use the following notation.
\begin{itemize}
\item The set $D_S$ is the set of nodes at which all players in set $S$ have played directly: $D_S = \{ z \in Z : \vec d_i[z] = 1 \forall i \in S \}$. The set $D'_S = Z \setminus D_S$ is its complement.
\item $\vx$ is a random variable for the correlated strategy profile played by all players through the $T$ timesteps. That is, $\vx$ is a uniform sample from $\{ \vx^{(1)}, \dots, \vx^{(T)}\}$.
\item $\pi(S|\vx)$ is the probability that a terminal node from set $S$ is reached, given that the mediator plays $\vmu$ and the players play the (possibly correlated) strategy profile $\vx$. That is, $\pi(S|\vx) = \Pr_{z \sim (\vmu, \vx)} [z \in S]$.
\item $\tilde u_i(\vx) = u_i(\vx) + \E_{z \sim (\vmu, \vx)} q_i(z)$ is the expected utility for player $i$, including payment, under profile $(\vmu, \vx)$. 
\item $u_i(\vec y_i - \vx_i, \vx_{-i}) := u_i(\vec y_i, \vx_{-i}) - u_i(\vx_i, \vx_{-i})$ is player $i$'s advantage for playing $\vec y_i$ instead of $\vec x$. $\tilde u_i(\vec y_i - \vx_i, \vx_{-i})$ and $\pi(z|\vec y_i - \vx_i, \vx_{-i})$ are defined similarly.
\end{itemize}

Let $\eps = R(T)/T$. Then after $T$ timesteps, since the players are no-regret learners, their average joint strategy profile will be an $P \eps$-NFCCE of the extensive-form game with the payments added. 

Intuitively, the proof will go as follows. We will show that, for $P$ sufficiently large, each player's incentive to be direct will be {\em at least} as great as it would have been if everyone else were also direct, plus $\alpha$. Then it will follow from the fact that $\vmu$ is an equilibrium, and picking $\alpha \gg P \eps$,  that all players must therefore be direct.
 We first prove a lemma. Informally, the lemma states that, when any player $i$ deviates, all other players must be direct.

\begin{lemma}\label{le:bandit-steering1}
Let $z$ be any node with $\vec d_i[z] = 0$, that is, any node at which player $i$ has deviated. Then
$\abs{\pi(z|\vx_i, \vec d_{-i} - \vx_{-i})}\le \gamma := n\eps + \sum_j\delta_j/P$, where $\delta_j := u_j(\vx_j - \vec d_j, \vec x_{-j})$ is player $j$'s current deviation benefit.
\end{lemma}
\begin{proof}
Assume without loss of generality that $i = 1$, and consider two cases.
\begin{enumerate}
\item $\vec d_{j}[z] = 0$ for some $j \ne i$---that is, some other player has also deviated. Then $ \pi(z|\vx_i, \vec d_{-i}) = 0$. Assume for contradiction that $ \pi(z|\vx) > \gamma$. Let $h_i, h_j \prec z$ be the two deviation points---that is, $\vec d_i[h_i] = 1$ but $\vec d_i[h_i a_i] = 0$ where $h_i a_i \preceq z$, and similar for $h_j$. Suppose without loss that $h_i \prec h_j$. Now consider player $j$'s incentive. If player $j$ were to switch to playing $\vec d_j$, its expected payment increases by at least $\gamma P$, and its expected utility (sans payment) decreases by $\delta_j$, by definition. When $\gamma \ge \eps + \delta_j/P$, this produces a contradiction.
\item $\vec d_j[z] = 1$ for all $j \ne i$. Then $\pi(z|\vx_i, \vec d_{-i}) \ge \pi(z|\vx)$, so we need to show that $\pi(z|\vx_i, \vec d_{-i}) - \pi(z|\vx^{(t)}) \le \gamma$. That is, other players will almost always play to catch player $i$ deviating, whenever possible. Suppose not. Let $h \prec z$ be the point where player $i$ deviated (that is, $\vec d_i[h] = 1$ but $\vec d_i [ha_1] = 0$ where $ha_1 \preceq z$). Let $a_0$ be the direct action at $h$. Notice that, for any player $j \ne i$, if $j$ shifts to playing the direct strategy, the probability of leaving the path to $ha$ before reaching $ha$ itself cannot increase by more than $\eps + \delta_i/P$: otherwise, player $j$'s expected utility would be increasing by more than $\delta_i$, a contradiction. If all $n-1$ players allocate their deviations in this manner, and even if the remaining $(n-1) \delta_i/P$ probability of leaving path $ha$ is then all allocated to node $z$, the reach probability of $z$ could not have increased by more than $\sum_j (\eps + \delta_j / P)$. Thus, when $\gamma$ is larger than this value, we have a contradiction. \qedhere
\end{enumerate}
\end{proof}
The rest of the proof is structured as follows. We will first show, roughly speaking, that {\em player $i$'s deviation benefit}---that is, its advantage for playing $\vx_i^{(t)}$ at each timestep $t$ instead of playing $\vec d_i$---is {\em smaller} against the opponent strategies $\vx_{-i}^{(t)}$ than it would be against $\vec d_{-i}^{(t)}$, modulo a small additive error. Then, the proof will follow from the fact that $\vec d$ is an equilibrium against $\vmu$, so therefore all players should play according to $\vec{d}$.
\begin{align}
 &  \tilde u_i(\vx) - \tilde u_i(\vx_i, \vec d_{-i})
 \\&= \sum_{z \in D_i \cap D_{-i}} \tilde u_i(z) \qty[\pi(z|\vx) - \pi(z|\vx_i, \vec d_{-i})] 
 + \sum_{z \in D_i \cap D'_{-i}} \tilde u_i(z) \pi(z|\vx)
 \\&\phantom{=}+ \underbrace{\sum_{z \in D'_i} u_i(z) \qty[\pi(z|\vx) - \pi(z|\vx_i, \vec d_{-i})] }_{\le \gamma |Z|}
\\&\le \sum_{z \in D_i \cap D_{-i}} \tilde u_i(z) \qty[\pi(z|\vx) - \pi(z|\vx_i, \vec d_{-i})] 
 + \sum_{z \in D_i \cap D'_{-i}} \tilde u_i(z) \pi(z|\vx) + \gamma |Z|
\end{align}
where we use, in order, the definition of expected utility, the fact that $u_i(z) = \tilde u_i(z)$ when $\vec d_i[z] = 0$ and $\pi(z|\vx) = 0$ whenever $\vx_i[z] = 0$ for any $i$, and finally \Cref{le:bandit-steering1}. Similarly,
\begin{align}
 &  \tilde u_i(\vec d_i, \vx_{-i}) - \tilde u_i(\vec d)
 \\&= \sum_{z \in D_i \cap D_{-i}} \tilde u_i(z) \qty[\pi(z|\vec d_i, \vx_{-i}) - \pi(z|\vec d)] 
 + \sum_{z \in D_i \cap D'_{-i}} \tilde u_i(z)  \pi(z|\vec d_i, \vx_{-i}).
\end{align}
Thus,
\begin{align}
 &P\eps -  \qty[ \tilde u_i(\vec d) - \tilde u_i(\vx_i, \vec d_{-i})] 
 \\&\ge \qty[ \tilde u_i(\vec d_i, \vx_{-i}) - \tilde u_i(\vx) ] -  \qty[ \tilde u_i(\vec d) - \tilde u_i(\vx_i, \vec d_{-i})]
 \\&\ge \sum_{z \in D_i \cap D_{-i}} \tilde u_i(z) \underbrace{ \qty[\pi(z|\vec d_i - \vx_i, \vx_{-i}) - \pi(z|\vec d_i - \vx_i, \vec d_{-i})] }_{\le 0}
  \\&\phantom=+ 2 \sum_{z \in D_i \cap D'_{-i}} \pi(z|\vec d_i - \vx_i, \vx_{-i})  - \gamma |Z|
\\&\ge 2 \sum_{z \in D_i \cap D_{-i}}  \qty[\pi(z|\vec d_i - \vx_i, \vx_{-i}) - \pi(z|\vec d_i - \vx_i, \vec d_{-i})]
\\&\phantom=+ 2 \sum_{z \in D_i \cap D'_{-i}} \pi(z|\vec d_i - \vec x_i, \vx_{-i})    - \gamma |Z|
\\&= 2 \qty[ \pi(D'_i|\vec x_i, \vec d_{-i}) - \pi(D'_i|\vec x) ]- \gamma |Z| \ge - 3\gamma |Z|.
\end{align}
The first inequality uses the fact $\pi(z|\vec d_i, \vx_{-i}) - \pi(z|\vx) \ge 0$ when $\vec d_i[z] = 1$ and $\tilde u_i(z) \ge P \ge 2$ when $\vec d_i[z] = 1$ and $\vec d_{-i}[z] = 0$. The quantity in braces is nonpositive because for any profile $\vx$, setting $\vx_{-i} = \vec d$ only increases the probability that player $i$ is the one to deviate from the path to $z$. The second inequality uses the nonpositivity of the quantity in the braces, and the fact that $\tilde u_i(z) = u_i(z) + \alpha \le 2$. 

Now we look at the remaining quantity, $ \tilde u_i(\vec d) - \tilde u_i(\vx_i, \vec d_{-i})$, which is simply the negative deviation of benefit of Player $i$'s strategy $ \vec x_i$ if all other players were direct. Indeed, since we know that $\vmu$ is an equilibrium, we have
\begin{align}
   & \tilde u_i(\vec d) - \tilde u_i(\vx_i, \vec d_{-i})
    \notag \\&= \underbrace{  \qty[ \tilde u_i(\vec d) - u_i(\vec d) ] }_{= \alpha}  - \underbrace{  \qty[ \tilde u_i(\vx_i, \vec d_{-i})   -  u_i(\vx_i, \vec d_{-i}) ]}_{= \alpha (1 -  \Delta_i(\vx_i, \vec d_{-i}))} +\underbrace{  \qty[ u_i(\vec d) -  u_i(\vx_i, \vec d_{-i})]  }_{\ge 0}
    \notag \\&\ge \alpha \Delta_i(\vx_i, \vec d_{-i})
    \ge \alpha  \Delta_i(\vx) - \gamma |Z|,\label{eq:bareps-inject}
\end{align}
where the final inequality again uses \Cref{le:bandit-steering1} and $\Delta_i(\vec x) := \sum_{z : \vec d_i[z] = 0} \pi(z|\vec x)$

Now, notice that $\delta_i \le \Delta_i(\vec x)$, by definition. Substituting into the previous inequality and \Cref{le:bandit-steering1}, we have
\begin{align}
    \alpha  \Delta_i(\vx) - 4\qty(n \eps + \frac{\sum_j \Delta_j(\vec x)}{P})|Z| \le P \eps,
\end{align}
or, rearranged,
\begin{align}
    \alpha  \Delta_i(\vx) - 4|Z| \frac{\sum_j \Delta_j(\vec x)}{P} \le (P + 4n) \eps \le 2P\eps
\end{align}
when $P \ge 4n$. Summing over all players $i$ yields
\begin{align}
    \alpha  \Delta - 4|Z|\frac{\Delta}{P} \le (P + 4n) \eps \le 2P\eps
\end{align}
where $\Delta = \sum_i \Delta_i(\vec x)$, or, rearranging,
\begin{align}
    \Delta \le \frac{2P\eps}{\alpha - 4|Z|/P}.
\end{align}
Both the payments from the mediator and the gap to optimal value are thus bounded by 
\begin{align}
    \alpha + P \Delta \le \alpha + \frac{2P^2\eps}{\alpha - 4|Z|/P}.
\end{align}
Now taking $\alpha = 4|Z|^{1/2} \eps^{1/4}$ and $P = 2|Z|^{1/2} / \eps^{1/4} $ gives the desired bounds.

\section{Trajectory-Feedback Online Steering in Normal-Form Games}
\label{app:bandit-online-nf}

Essentially, the algorithm replicates the {\sc OnlineSteer} algorithm (\Cref{def:online-steer}) by randomly sampling. In the normal-form setting, a mediator pure strategy is a profile of actions, $d^{(t)} \in A_1 \times \dots \times A_n$, where $A_i$ is the action set of player $i$. Each player $i$ observes the recommendation $d_i^{(t)}$, and chooses an action $a_i^{(t)}$. 

\begin{definition}[{\sc NormalFormSteer}]
    The mediator runs a {\em bandit} regret minimization algorithm $\mc R_\mediator$, such as Exp3~\cite{Auer02:Nonstochastic}, over its own strategy space $X_0$, which we assume has regret at most $R_\mediator(T)$ after $T$ rounds. On each round, the mediator does the following.
    \begin{enumerate}
        \item Get a strategy $d^{(t)} = (d_1^{(t)}, \dots, d_n^{(t)})$ from $\mc R_\mediator$. 
        \item With probability $\alpha$, ignore $d^{(t)}$ and recommend actions $(\tilde d_1^{(t)}, \dots, \tilde d_n^{(t)})$ uniformly at random. Let $a^{(t)} = (a_1^{(t)}, \dots, a_n^{(t)})$ be the tuple of actions played by the players. Pay each player 
        \begin{align}
            q_i^{(t)}(a^{(t)}) := 1 - u_i(a^{(t)}) + \ind{a_i^{(t)} = \tilde d_i^{(t)}}.
        \end{align}
        Pass reward $0$ to the mediator.
        \item Otherwise, give recommendation $r_i^{(t)}$ to each player $i$. Pay each player 
        \begin{align}
            q_i^{(t)}(a^{(t)}) := u_i(a_i^{(t)}, d_{-i}^{(t)}) - u_i(a^{(t)}) - \min_{a_i' \in A_i} \qty[u_i(a_i', d_{-i}^{(t)}) - u_i(a_i', a_{-i}^{(t)})].
        \end{align}
        Pass reward $\frac{1}{\lambda} u_0(d) - \sum_{i=1}^n \qty[ u_i(a_i^{(t)}, d_{-i}) - u_i(d) ]$ to the mediator.
    \end{enumerate}
\end{definition}

\begin{theorem}\label{th:nf-online-steer}
     Set the hyperparameters $
        \alpha = |Z|^{1/3} n^{-2/3} b^{1/3} \eps^{2/3}  $ and $ \lambda = |Z|^{1/3} n^{1/3} b^{1/3} \eps^{-1/3}
    $
    where $\eps := (R_\mediator(T) + 4nR(T)) / T$ is the average regret bound summed across players and $b = \max_i |A_i|$. Let $T$ be large enough that $\alpha \le 1/(2n)$.
    Then running {\sc NormalFormSteer} results in average realized payments, directness gap, and optimality gap all bounded by $10 \lambda^* |Z|^{4/3} \eps^{1/3}$.
\end{theorem}
\begin{proof}
    Reverting to the extensive-form notation, the expected utility of the mediator on iteration $t$ is 
    \begin{align}
       (1 - \alpha) \qty(\frac{1}{\lambda} u_0(\vec\mu^{(t)}, \vec d) - \sum_{i=1}^n \qty[ u_i(\vec\mu^{(t)}, \vec x_i^{(t)}, \vec d_{-i}) - u_i(\vec\mu^{(t)}, \vec d)]).
    \end{align}
    The expected utility of player $i$ is, up to an additive term that cannot be affected by player $i$,
    \begin{align}
       \alpha\frac{1}{|A_i|} \vec d_i^\top \vec x_i + (1 - \alpha) \qty(u_i(\vec \mu^{(t)}, \vec x^{(t)}, \vec d_{-i}) - u_i(\vec \mu^{(t)}, \vec d^{(t)}, \vec d_{-i}) ).
    \end{align}
    Therefore, the players and mediator experience the same utilities that they would in the zero-sum game 
    \begin{align}\label{eq:nf-online-steering-game}
        \max_{\vec\mu \in \Xi} \min_{\vec y \in Y}~ (1 - \alpha) \qty(\frac{1}{\lambda} \vec c^\top \vec\mu - \vec\mu^\top \mat A \vec y) - \alpha \sum_i \frac{1}{|A_i|}\vec d^\top_i  \vec y_i,
\end{align}
where, as in the proof of \Cref{th:online-steer},  $\vec y := \vec x - \vec d$. Following the proof of \Cref{th:online-steer}, we conclude that $(\vec{\bar\mu}, \vec{\bar y})$ must be an $\eps$-Nash equilibrium of the above zero-sum game. Let $\lambda^*, \lambda, \vec y^*, \vec y'$ be as in that proof. For simplicity of notation, let $\vec D$ be the vector satisfying $\vec D^\top \vec y = \sum_i \frac{1}{|A_i|}\vec d^\top_i  \vec y_i,$ Then
\begin{align}
-\alpha \vec D^\top \vec{\bar y} \le 
    \max_{\vec\mu \in \Xi} \min_{\vec y \in Y}~ (1 - \alpha) \qty(\frac{1}{\lambda} \vec c^\top \vec\mu - \vec\mu^\top \mat A \vec y) - \alpha \vec D^\top \vec y \le -\alpha \vec D^\top \vec y' + \eps
\end{align}
or, rearranging,\footnote{We note once again that $-\vec d^\top \bar{\vec y}$ and $-\vec D^\top \bar{\vec y}$ are, despite the negative sign, a {\em nonnegative} quantities since $\vec y = \vec x - \vec d$.}
    \begin{align}
       -\frac{1}{b} \vec d^\top \bar{\vec y} \le  -\vec D^\top \bar{\vec y} \le -\frac{\lambda^*}{\lambda}\vec D^\top \vec y^* + \frac{\eps}{\alpha} \le \frac{\lambda^*}{\lambda}n + \frac{\eps}{\alpha} := \frac{\delta}{b}.
    \end{align}
    where $b = \max_i |A_i|$ is the maximum branching factor.
    Thus, by \Cref{lem:obedience-union}, the average payment is bounded by $|Z|(2\delta + \alpha)$.
    We now turn to the mediator's average utility. The equilibrium value of \eqref{eq:nf-online-steering-game} is at least $-\alpha n$ (achieved by the optimal equilibrium), in turn implying that the current value in the game under $(\Bar{\vmu}, \Bar{\vec{y}})$ is at least $-\alpha n - \eps$. So,
    \begin{align}
        \E_{t \in \range{T}} u_\mediator(\vmu^{(t)}, \vec d) = \vec c^\top \bar{\vmu}\ge \min_{\vec y \in Y}~ \vec c^\top \bar{\vec\mu} - \lambda \qty[\bar{\vec\mu}^\top \mat A \vec y - \frac{\alpha}{1 - \alpha} \vec d^\top \vec y] \ge -2\lambda(\alpha n + 2\eps).
    \end{align}
    since $\alpha \le 1/2$.
    By \Cref{lem:obedience-union} again,
\begin{align}
    \abs{\E_{t \in \range{T}} u_\mediator(\vec\mu^{(t)}, \vec x^{(t)}) - u_\mediator(\vec\mu^{(t)}, \vec d)} \le \norm{\E_{t \in \range{T}} \hat{\vec x}^{(t)} - \hat{\vec d}}_1 \le \E_{t \in \range{T}} \norm{\hat{\vx}^{(t)} - \hat{\vec d}}_1 \le |Z|\delta,
\end{align}
so the optimality gap is bounded by $4\eps\lambda + 2n \alpha \lambda  + |Z| \delta$, and the directness gap is bounded by $|Z|\delta$. It thus suffices to select hyperparameters $\alpha$ and $\lambda$ so as to minimize the following expression, which is an upper bound on all three gaps:
\begin{align}
  4\eps\lambda + 2n \alpha \lambda  + 2 |Z| \delta \le 4\eps\lambda + 2n \alpha\lambda + 2|Z|nb \frac{\lambda^*}{\lambda} + 2|Z|b \frac{\eps}{\alpha}.
\end{align}
In particular, setting the hyperparameters
\begin{align}
    \alpha = |Z|^{1/3} n^{-2/3} b^{1/3} \eps^{2/3}  \qq{and} \lambda = |Z|^{1/3} n^{1/3} b^{1/3} \eps^{-1/3}
\end{align} we arrive at the bound
\begin{align}
    4\eps^{2/3}(|Z|nb)^{1/3} + 2 (|Z|nb)^{2/3} \eps^{-1/3} + 2  \lambda^* (|Z|nb)^{2/3} \eps^{-1/3} + 2 (|Z|nb)^{2/3} \eps^{-1/3} \le 10  \lambda^* |Z|^{4/3}\eps^{1/3},
\end{align}
as claimed.
\end{proof}

\section{Further Experimental Results}\label{app:experiments}

Here, we provide plots akin to those in \Cref{fig:body-experiments} for other games and solution concepts.
For a description of the solution concepts used in these plots, see \citet{Zhang22:Polynomial}.
We experiment on four standard benchmark games, which are the same ones used in by~\citet{Zhang23:Computing}.
\begin{itemize}
    \item \textbf{Kuhn poker}. We use the three-player version of this standard benchmark introduced by \citet{Kuhn50:Simplified}.
    
    \item \textbf{Sheriff}. This game, introduced as a benchmark for correlation in extensive-form games by \citet{Farina19:Correlation}, is a simplified version of the \emph{Sheriff of Nottingham} board game. A \emph{Smuggler}---who is trying to smuggle illegal items in their cargo---and the \emph{Sheriff}---whose goal is stopping the Smuggler. Further details on the game can be found in \citet{Farina19:Correlation}.

The Smuggler first chooses a number $n\in\{0,1\}$ of illegal items to load on the cargo. Then, the Sheriff decides whether to inspect the cargo. If they choose to inspect, and find illegal goods, the Smuggler has to pay $n$ to the Sheriff. Otherwise, the Sheriff compensates the Smuggler with a reward of $1$. If the Sheriff decides not to inspect the cargo, the Sheriff's utility is 0, and the Smuggler's utility is $5n$.
After the Smuggler has loaded the cargo, and before the Sheriff decides whether to inspect, the Smuggler can attempt to bribe the Sheriff. To do so, they engage in $2$ rounds of bargaining and, for each round $i$, the Smuggler proposes a bribe $b_i\in\{0,1,2\}$, and the Sheriff accepts or declines it. Only the proposal and response from the final round are executed. If the Sheriff accepts a bribe $b_2$ then they get $b_2$, while the Smuggler's utility is $5n-b_2$.

    \item \textbf{Battleship}. This game, introduced as a benchmark for correlation in extensive-form games by \citet{Farina19:Correlation}, is a general-sum version of the classic game Battleship, where two players take turns placing ships of varying sizes and values on two separate grids of size $2\times 2$, and then take turns firing at their opponent. Ships which have been hit at all their tiles are considered destroyed. The game ends when one player loses all their ships, or after each player has fired $2$ shots. Each player's payoff is determined by the sum of the value of the opponent's destroyed ships minus two times the number of their own lost ships. 

    \item \textbf{Ridesharing}. A benchmark introduced in \citet{Zhang22:Optimal}. Two drivers compete to serve requests on a road network, an undirected graph $\Grs=(\Vrs,\Ers)$ depicted in \Cref{fig:map} with unit edge cost. Each vertex $v\in \Vrs$ corresponds to a ride request to be served. Each request has a reward in $\mathbb{R}_{\ge 0}$, which is shown in set notation at vertices in the graph. The first driver arriving at node $v\in \Vrs$ serves the ride and receives the associated reward. The game terminates when all nodes have been cleared, or after $T=2$. If the two drivers arrive simultaneously on the same vertex they both get reward 0. Final driver utility is computed as the sum of the rewards obtained from the beginning until the end of the game.

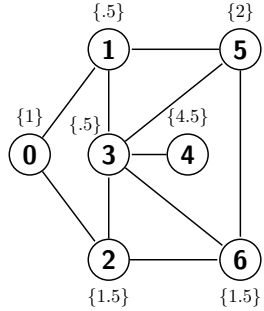
\begin{figure}[htp]
\centering
\scalebox{.7}{
\begin{tikzpicture}[shorten >=1pt,auto,node distance=3cm,
  thick,main node/.style={circle,draw,
  font=\sffamily\Large\bfseries,minimum size=5mm}]

  \node[main node, label={[align=left]\{1\}}] at (0, 0)   (0) {0};
  \node[main node, label={[align=left]\{.5\}}] at (1.5, 2)   (1) {1};
  \node[main node, label={[left,label distance=.15cm]\{.5\}}] at (1.5, 0)   (2) {3};
  \node[main node, label={[anchor=north]below:\{1.5\}}] at (1.5, -2)   (3) {2};
  \node[main node, label={[align=left]\{4.5\}}] at (3, 0)   (4) {4};
  \node[main node, label={[align=left]\{2\}}] at (4, 2)   (5) {5};
  \node[main node, label={[anchor=north]below:\{1.5\}}] at (4, -2)   (6) {6};

  \path[every node/.style={font=\sffamily\small,
  		fill=white,inner sep=1pt}]

    (0) edge[-] (1)
    (0) edge[-] (3)
    (1) edge[-] (2)
    (3) edge[-] (2)
    (4) edge[-] (2)
    (3) edge[-] (6)
    (5) edge[-] (6)
    (5) edge[-] (1)
    (2) edge[-] (5)
    (2) edge[-] (6);
\end{tikzpicture}}
\caption{Map used in the ridesharing game. Rewards are in curly braces.}
\label{fig:map}
\end{figure}

\end{itemize}
For the following results, we use a burn-in of 10 iterates (that is, no payments
are issued in the first 10 iterations; steering only begins after that).

For each game, we consider the problem of steering the learners towards an optimal instance of each of the solution concepts. The objective function used to define optimality is set to be social welfare for general-sum games, and the utility of Player~1 for the three-player zero-sum game (Kuhn poker). For each combination of game and equilibrium concept, we show four plots. Each corresponding to a different value of the payment multiplyer $P_\text{mult} \in \{1,2,4,8\}$. The payment multiplier controls the value of $P$, which is set to $P \defeq P_\text{mult} \times $ the reward range of the game.

We observe that in all games and equilibrium concepts, our algorithm is able to steer the learners towards the optimal social objective, as predicted by our theory. As $P_\text{mult}$ grows, we observe that the convergence speed increases, at the cost of a higher payment magnitude.

{
\newcommand{\addplot}[3]{%
    \def\file{experiments/burn_in_10/#1/P_#2/#3.pdf}%
    \IfFileExists{\file}{%
        \includegraphics[width=.4\textwidth]{\file}
    }{%
        \tikz\node[draw,minimum width=6cm,minimum height=5.0cm] {Missing plot (#1, \MakeUppercase{#3})};%
    }%
}

\newcommand{\plotGameEq}[2]{%
    \begin{tabular}{ll}%
        \multicolumn{2}{c}{\fbox{Solution concept: \MakeUppercase{#2}}}\\
        \addplot{#1}{1}{#2}&%
        \addplot{#1}{2}{#2}\\
        \addplot{#1}{4}{#2}&%
        \addplot{#1}{8}{#2}
    \end{tabular}%
}%

\newcommand{\plotGame}[2]{%
    \subsection{Game: #2}
    \begin{center}
        \plotGameEq{#1}{efce}\\%
        \plotGameEq{#1}{efcce}\\%
        \plotGameEq{#1}{nfcce}\\%
        \plotGameEq{#1}{nfccert}\\%
        \plotGameEq{#1}{ccert}\\%
        \plotGameEq{#1}{cert}\\%
        \plotGameEq{#1}{comm}%
    \end{center}
}%

\plotGame{K35}{Kuhn Poker}
\plotGame{S2122}{Sheriff}
\plotGame{U212}{Ridesharing}
\plotGame{B2222}{Battleship}
}

\end{document}